\documentclass[journal,10pt]{IEEEtran}
\usepackage{algpseudocode}
\usepackage{algorithm}
\usepackage{algpseudocode}
\usepackage{amsmath}
\usepackage{multirow}
\usepackage{graphicx}
\usepackage{caption}
\usepackage{booktabs}
\usepackage{multirow}
\usepackage{url}
\usepackage{amsfonts}
\usepackage{amssymb}
\usepackage{graphicx} %use graph format
\usepackage{subfigure}
\usepackage{caption}
\usepackage{color}

\usepackage{cite}

\newtheorem{proof}{ \textbf{Proof}}
\newtheorem{theorem}{ \textbf{Theorem}}

\newtheorem{remark}{ \textbf{Remark}}

\begin{document}

\title{A Mobility-Aware Vehicular Caching Scheme in Content Centric Networks: Model and Optimization}

\author{Yao~Zhang, Changle~Li,~\IEEEmembership{Senior Member,~IEEE}, Tom H. Luan,~\IEEEmembership{Senior Member,~IEEE}, Yuchuan~Fu, Weisong Shi,~\IEEEmembership{Fellow,~IEEE}
 and Lina Zhu~\IEEEmembership{Member,~IEEE}

%\thanks{Manuscript received June 22, 2018. This work was supported by the National Natural Science Foundation of China under %Grant No. 61571350.}
\thanks{Copyright (c) 2015 IEEE. Personal use of this material is permitted. However, permission to use this material for any other purposes must be obtained from the IEEE by sending a request to pubs-permissions@ieee.org.
\emph{(Corresponding author: Changle Li.)}}
\thanks{This work was supported in part by the National Natural
Science Foundation of China under Grants U1801266, 61571350 and 61601344, and in part by the Key Research and Development Program of Shaanxi under Contract 2017KW-004, Contract 2017ZDXM-GY-022, Contract 2018ZDXMGY-038, and Contract 2018ZDCXL-GY-04-02.}
\thanks{Y. Zhang, C. Li, Y. Fu, and L. Zhu are with the State Key Laboratory of Integrated Services Networks, Xidian University, Xi'an 710071, China
(e-mail: clli@mail.xidian.edu.cn).}
\thanks{T. H. Luan is with the School of Cyber Engineering, Xidian University,
Xi'an 710071, China (e-mail: tom.luan@xidian.edu.cn).}
\thanks{W. Shi is with the Department of Computer Science, Wayne State University, Detroit, MI 48202 (e-mail: weisong@wayne.edu).}
}

%\markboth{IEEE Transactions on Vehicular Technology,~Vol.~XX, No.~XX, XXX~2017}
{}
%{Shell \MakeLowercase{\textit{et al.}}: Bare Demo of IEEEtran.cls for Journals}
\maketitle

\begin{abstract}

Edge caching is being explored as a promising technology to alleviate the network burden of cellular networks
by separating the computing
functionalities away from cellular base stations.
However, the service capability of existing caching scheme is limited by fixed edge infrastructure when facing the
uncertainties of users' requests and locations. The vehicular caching, which uses the moving vehicles as cache carriers,
is regard as an efficient method to solve the problem above.
This paper studies the effectiveness of vehicular caching scheme in content centric networks by developing optimization
model towards the minimization of network energy consumption. Particularly,
we model the interactions between caching vehicles and mobile users as a $2$-D Markov process, in order to characterize
the network availability of mobile users. Based on the developed model, we propose an online vehicular caching design by optimizing network energy efficiency. Specifically, the problem of caching decision making is firstly formulated as a fractional optimization
model, towards the optimal energy efficiency. Using nonlinear fractional programming technology and Lyapunov optimization theory, we
derive the theoretical solution for the optimization model. An online caching algorithm to enable the optimal vehicular caching
is developed based on the solution. Finally, extensive simulations are conducted to examine the performance of our
proposal. By comparison, our online caching scheme outperforms the existing scheme in terms of
energy efficiency, hit ratio, cache utilization, and system gain.

\end{abstract}


\IEEEpeerreviewmaketitle

\section{Introduction}\label{SectionI}

The explosive growth of mobile data \cite{indexglobal} is
driving a myriad of novel services and applications, such as augmented/virtual reality, ultra-high definition video,
which make mobile users enjoy a fairly rich network experience. However,
the resulting data tsunami seriously challenges mobile operators worldwide in
their network performance \cite{yu2017mobile}, \textit{e.g.},
network capacity, Quality of Service (QoS), and energy consumption. Particularly,
a predication report from Cisco shows that there are about $49$ exabytes of monthly global mobile data with $11.6$ billion
of mobile devices in $2021$, increasing about sevenfold between $2016$ and $2021$ \cite{indexglobal}. The large-scale deployment of infrastructure is the simplest
method to accommodate titanic data requests in the future, which however incurs huge monetary costs.

A variety of novel and exciting techniques
have been studied to tackle the problem. By deploying
large-scale small cells, UDN (Ultra-Dense Network) has been stressed upon as a
promising technology for the next generation of mobile communication, aiming
to reduce the burden at MBSs (Macro Base Stations) \cite{chopra2018possible}.
A new computing paradigm called edge computing has been proposed with the push from cloud services and pull from
the voluminous and complex data, in order to provide the reliable data
processing at the network edge \cite{Shi2016Edge, hui2017content}. Data offloading
is an efficient method to improve the network experience of mobile users by
using the complementary technologies in $5$G systems, such as Wi-Fi, to
offload the mobile data originally targeted toward cellular networks
\cite{yu2017mobile}. By offering storage resource to the edge of network, the
edge caching is recently proposed to deal with the increasing data demand of
mobile users and balance the overload in cellular networks, in order to satisfy
the ultra-low latency requirement of next generation mobile networks
\cite{chen2018data}.

All above technologies focus on scheduling computation and communication
resources to achieve a better system utilization, targeting on the
improvement of network capacity.
Nevertheless, they fail to consider the uncertainties of mobile users, including locations and
requests, which will degrade the QoE (Quality of Experience) perceived by users due to the fixed service range and low utilization of mobile data.
Besides, the strong dependence on fixed infrastructure also limit the extension of service coverage.
The vehicular caching, which caches mobile data in vehicles ($e.g.$, taxis and buses) and uses the mobility of vehicles to improve the service range and capacity of caching, is a potential method.

In this work, we explore the energy-efficient caching services provided by vehicular
caching scheme in content-centric networks \cite{jacobson2009networking} for mobile users. The content-centric networking (CCN) is a communication architecture built on named data \cite{jacobson2009networking}. Compared with traditional networks that built on named hosts, CCN requires less backhaul energy at the cost of
caching energy in wireless caching
networks, which thus contributes to the construction of energy-efficient wireless networks.
On the other hand, as a key component of $5$G network, the
vehicular communication techniques enabled by IEEE 802.11p or LTE (Long Term Evolution) have been widely studied in recent years, including
communication protocols \cite{zhu2016geographic, zhu2015on, wang2012multi},
multi-hop backbone communications \cite{li2018building}, big date driven vehicular social networks
\cite{cheng2018bigdate, luan2015social}, and traffic safety \cite{fu2018infrastructure}. Most of these researches aim to deal with the highly-dynamic network topology caused by
the mobility of vehicles. The extensive interactions between moving vehicles and mobile users provide a positive condition to improve the
service range of vehicular caching scheme. However, this natural feature has not yet attracted wide attention.
Therefore, using the mobility of vehicles to provide mobile users with fast and reliable data access in CCN becomes our motivation.

Different from the existing works above, three advantages is achieved by the mobility-aware vehicular caching
scheme in CCN. 1) \emph{Cost-saving deployments}: By caching the mobile data in moving vehicles, the large-scale
deployment of infrastructure is reduced. The cost
input in both operators and users thus would be significantly decreased. 2) \emph{Enhanced service capacity}: Using
the mobility of caching vehicles, the service capacity of caching is enhanced, including the service range and cache utilization. This is because that users will encounter multiple caching vehicles in a short time, thus the caching vehicles
may perform multiple times services after caching data one time. 3) \emph{Energy-efficient updating}: Based on CCN,
the cached mobile data is managed by naming information, the overhead of cache updating thus can be reduced. Besides, real-time
V2V communications can be used to share the caching data, which also decreases the overhead of backhaul.
To achieve above goals, the optimal caching decision in caching vehicles becomes the major object. In this
paper, we first explore the relationship of caching vehicles and mobile users. Based on the interactions between moving vehicles
and mobile users, a vehicular caching design in CCN is proposed to improve the network energy efficiency.
%we first model the interactions between moving vehicles and mobile users by taking advantage of both the
%mobility of vehicles and the request features of mobile users. After that, a vehicular caching scheme that explores what
%and where to cache mobile data among moving vehicles in CCN is proposed, in order to provide better caching services for mobile users and alleviate the burden
%of MBSs. Besides, the energy consumption, which is a critical performance metric in wireless cellular networks, will also play an important
%role in vehicular networks with the development of electric vehicles. As such, to determine the optimal
%caching decision, we design an optimization algorithm based on the fractional
%programming technology and Lyapunov optimization theory, targeting on optimizing the network energy
%efficiency.
We proceed in three steps:
%Our main contributions are summarized as follows:

\begin{itemize}
\item \emph{Modeling}: We formulate the interactions between caching vehicles and
mobile users as a $2$-D Markov process, in order to
characterize the network availability of mobile users. This is the first step to
incorporate the mobility of vehicles, and to provide more flexible and wider caching services for mobile users.
Based on this model, the service probability of mobile users can be obtained.

\item \emph{Designing}: Based on the developed model above, we propose a vehicular
caching scheme by caching mobile data in vehicles.
To make the optimal caching decision, we first
formulate the network energy efficiency in vehicular caching as a fractional optimization problem, and then explore the solution by incorporating fractional programming technology and Lyapunov optimization theory.
Based on the solution, a novel online algorithm is proposed
to ensure the energy efficiency oriented vehicular caching.
%, complementing the centralized cellular architecture
%to serve the requests of mobile users. By using the fast-moving and widely-distributed vehicles as data
%carrier, the service range can be largely extended with small cost. Besides, the data sharing by V2V communications
%alleviates the burden of MBS.

%\item \emph{Algorithm}: To determine the optimal caching decision, we first
%formulate a fractional optimization problem with the target of energy efficiency,
%and then explore the solution by incorporating fractional programming technology and Lyapunov optimization theory.
%Based on the solution, a novel online algorithm is proposed
%to ensure the energy efficiency oriented vehicular caching.

\item \emph{Validations}: Finally, our proposal is evaluated by extensive
simulations. Simulation results show that our scheme achieves a better
performance in terms of energy efficiency, hit ratio, cache utilization and system gain.
\end{itemize}

The remainder of this paper is structured as follows: Section \ref{Section II} presents a
briefly survey about the existing works related to our study. Section \ref{Section III}
illustrates the system model of our research.
%Section \ref{Section IV} models the interactions between caching vehicles and
%mobile users.
Section \ref{Section V} firstly formulates the vehicular caching as an optimization problem and then
introduces the designed online vehicular caching scheme. Section \ref{Section VI}
conducts the performance evaluation about our proposal while Section \ref{Section VII} closes
our paper with conclusion.
%\hfill August 4, 2015

%\begin{figure*}[!ht]
%\centering
%\subfigure[Average speed in US-$101$]{
%\label{1speed distribution}
%\includegraphics[width=3.6in,height=2.4in]{1averagespeedofvehicle_us101.eps}
%}\hspace{-12ex}
%\subfigure[Average speed in I-$80$]{
%\label{2speed distribution}
%\includegraphics[width=3.6in,height=2.4in]{2averagespeedofvehicle_i80.eps}
%}
% \caption{Results about average speed}
% \label{speed distribution}
%\end{figure*}\vspace{-2ex}
%
%\begin{figure*}[!ht]
%\centering
%\subfigure[Truck distribution in US-$101$]{
%\label{truckratio_us101}
%\includegraphics[width=3.0in,height=2.2in]{truckratio_us101.eps}
%}\hspace{-0ex}
%\subfigure[Truck distribution in I-$80$]{
%\label{truckratio_I80}
%\includegraphics[width=3.0in,height=2.2in]{truckratio_i80.eps}
%}
% \caption{Results about truck distribution}
% \label{truckratio}
%\end{figure*}

\section{Related Works}\label{Section II}

Edge caching has been extensively studied in recent years. In this section, we briefly survey existing literature in edge caching from the
perspective of access networks \cite{liu2016energy}, coding \cite{gabry2016energy, ji2017order}, prediction \cite{song2017learning, zhao2018mobility}, and
vehicular communications \cite{vigneri2017per, Vigneri2017Quality}.

Liu \emph{et al.}\cite{liu2016energy} study the performance analysis of typical BS-assisted
caching networks. The factors that impact caching performance in cache-enabled wireless access
networks, including interference, backhaul capacity, BS density, and cache capacity, are investigated to efficiently deploy BS cache.
To improve the utilization of cache storage, coding based edge caching is a potential research direction in recent years. Gabry \emph{et al.} \cite{gabry2016energy} explore the
impact of MDS (Maximum-Distance Separable) coding on the energy efficiency performance of edge caching, in order to minimize the backhaul rate and
the total energy consumption. Ji
\emph{et al.} \cite{ji2017order} propose a coded distributed caching system for the canonical shared link caching network based on linear index
coding. They conclude that caching file fragments rather than full files will obtain a better performance.
\begin{table}[t]
\centering
 \caption{\label{TABLEParameters}}
 \begin{tabular}{ll}
  \toprule
  Notations & Descriptions \\
  \midrule
 $\digamma$ & Set of files \\
 $B$ & Fragment size\\
 $W$ & System bandwidth\\
 $q_{j}$ & Cache decision of $F_{j}$\\
 $p_{j}$ & Request probability of $F_{j}$\\
 $\phi$ & Zipf exponent\\
 $N_{f}$ & Number of fragments \\
 $N_{v}$ & Number of vehicles\\
 $N_{u}$ & Number of users\\
 $R$ &  Radius of cellular cell\\
 %$\rho_{v}$ & Density of vehicles\\
 $\mu$ & Mean service rate of cellular network\\
 $\nu$ & Mean service rate of caching vehicles\\
 $\lambda$ & Mean arrival rate of users' requests\\
 $\xi$ &  Mean inter-meeting rate between users and \\
 &caching vehicles \\%
 $\omega$ & Tolerant time of requests \\
 %$\rho_{v}$ & Vehicle density \\
 $\sigma^{2}$ & Noise power \\
% $\gamma$ &   \\
  \bottomrule
 \end{tabular}
\end{table}
Another important problem in caching is the uncertainty of users' requests, which poses a serious impact on hit ratio edge caching. Song \emph{et al.}
\cite{song2017learning} propose a MAB (multi-armed bandit) based content caching and sharing scheme, in order
to profile the unknown content popularity and make the caching more efficient. Focusing on the scenario of vehicular communications, Zhao \emph{et al.}
\cite{zhao2018mobility} design a multi-tier caching mechanism based on a novel hybrid Markov model to predict the connection of
vehicles and RSUs. As such, the content offloading
in RSU can be optimized.

All above researches focus on improving the performance of the edge caching scheme/algorithm. However, they fail
to consider the limited service range and cache utilization due to the strong dependence on fixed infrastructure.
Different from these works, we use the moving vehicles as cache carriers, called vehicular caching.
We shift vehicles from service consumers
to service providers, in order to achieve a higher flexibility to provide caching services. On the other hand, mobile users in
traditional caching schemes have to move to encounter different cache carriers, resulting in the degradation of QoE and also causing low hit ratio. This drawback can also be overcome by vehicular caching because caching vehicles can encounter and serve multiple mobile users continuously after updating the cache one times.

In recent years, Vigneri \emph{et al.} are devoted to the studies of vehicular caching \cite{vigneri2017per,
Vigneri2017Quality}. They highlight the strength of vehicular communications based on the comparison between the caching in vehicular networks and cellular networks. In this work, we further develop an online vehicular caching scheme and explore the improved performance in hit ratio, energy efficiency, cache utilization and system gain. Particularly, we explore the performance enhancement by adopting named data information and D2D communications in vehicular caching scheme, which is the first work to do this research to the best of our knowledge.

\section{System Model}\label{Section III}
This section elaborates on the system model of our study. We first describe the network scenario, and then introduce the
communication model based on D2D communications. Before modeling the interactions of caching vehicles and mobile users, we
present the energy consumption models to evaluate the impact of
vehicular vehicles on the total energy consumption. The notations in our analysis are listed in TABLE \ref{TABLEParameters}.

\subsection{Network Scenario}

\begin{figure}[t!]
\centering
\includegraphics[width=3.5in]{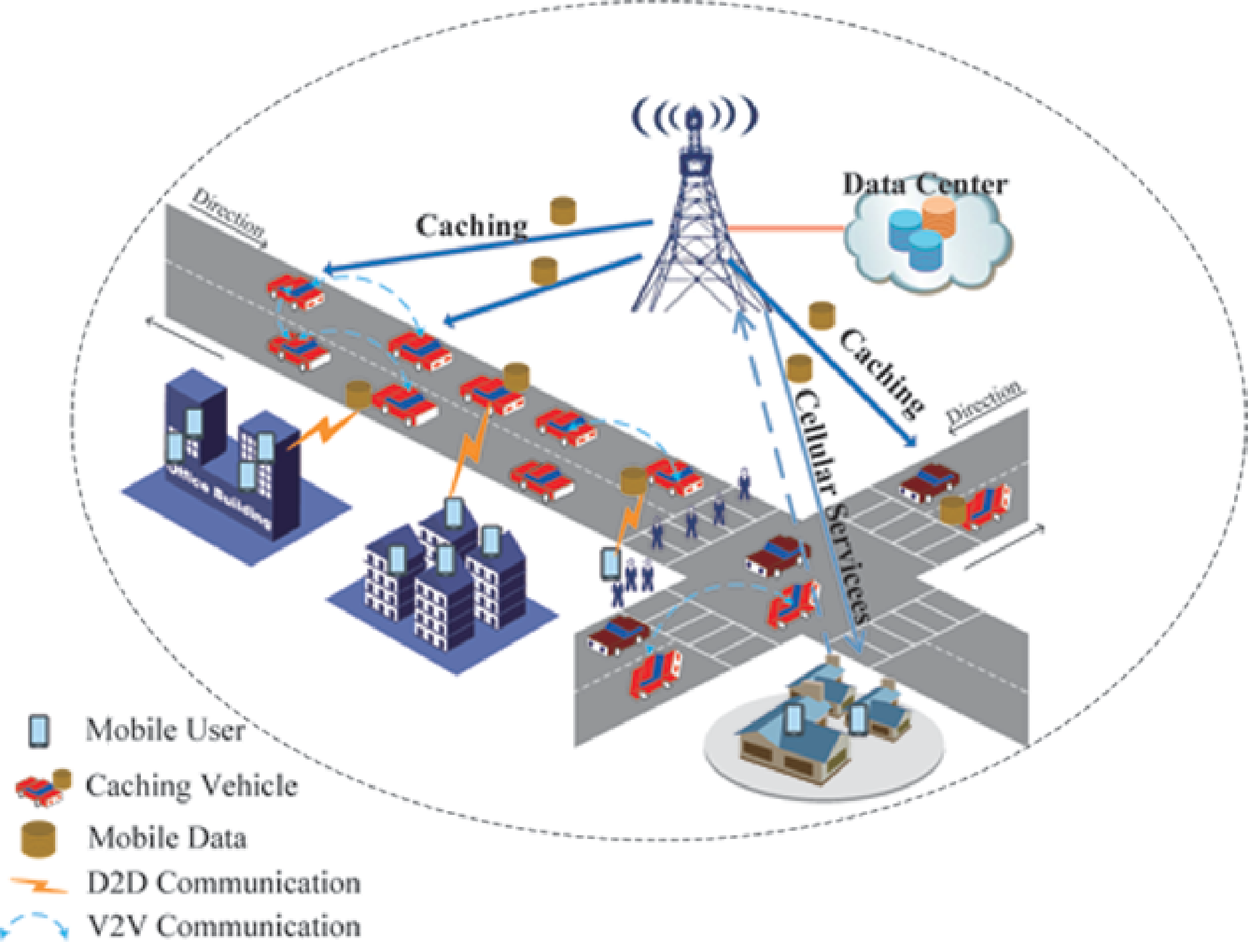}
\centering
\caption{Network Scenario}%
\label{SystemScenario}%
\end{figure}

In our study, we consider an urban scenario within the coverage of a cellular MBS,
as shown in Fig. \ref{SystemScenario}.
Different from traditional cellular networks,
we focus on an edge caching CCN where the requests of users are served by either MBS or caching
vehicles, in order to avoid the system performance degradation incurred by
the rapid increasing of mobile users and data demands.
%In traditional cellular networks, mobile users have to send real-time
%requests to MBS, in order to access the mobile services from cellular networks.
%To avoid the system performance degradation incurred by
%the rapid increasing of mobile users and data demands, we focus on an edge caching
%CCN where the requests of users are served by either MBS or caching
%vehicles.
Nowadays, the complex and large-scale road
systems make the meetings between mobile users and moving vehicles
frequent. As such, the real-time communications between mobile users and moving vehicles
can directly take place with the support of IEEE 802.11p or LTE protocol as long
as their distance is within the valid communication range. We apply D2D (Device-to-Device) communication technology to support
data transmissions between caching vehicles and mobile users because of its
strength in resource utilization and network throughput \cite{zhang2015interference}.
The communication modes of mobile users are determined by network environments.
Since this work focuses on the solution when or before the network burden takes place,
we assume that in this case the communications with caching vehicles act as
the preference setting of mobile users. Two types of data packets are involved in
CCNs, $i.e.,$ Interest packets and Data packets \cite{jacobson2009networking}.
The same content is identified by matching the names in both Interest packets and Data packets.
Note that the name of each packet is an opaque and binary object to networks. As such,
the sources of packets (in public naming conventions) even packet names (in private naming conventions)
cannot be obtained by other network members, such as routers, in order to protect the privacy of users
and transmission security \cite{jacobson2009networking}. Therefore, multiple users interested the
same content can share the transmissions from caching vehicles. Specifically,
each user sends Interest packets involving
the naming information of requested content to nearby vehicles by broadcasting.
Any vehicle that receives the Interest packets and has the requested content in the cache storage directly deliveries the Data packets to the user. Otherwise,
the vehicle sends these pending Interests to surrounding vehicles and informs the user about the pending time.
The pending time is determined by the distance between the vehicle caching requested content and the user.
If the pending time exceeds the user's tolerance time, the request will be switched to MBS.
Therefore, there are three types of communication nodes in our
scenario, \textit{i.e.}, MBS, mobile users, and caching vehicles. Mobile users
can access network services either from the always-on MBS by cellular communications or
from the opportunistic caching vehicles by D2D communications. It should be noted that the network
throughput is improved by caching mobile data in vehicles while the energy consumed by caching is accordingly increased. Since the network performance
varies over time, we start our analysis by assuming networks operate in
slotted time, \textit{i.e.}, time slot $t$ represents the time interval
$[t,t+1)$, $t\in\left\{  0,1,2,...\right\}  $. Other assumptations in
analysis, which are simple but can capture the basic elements, are shown as follows.

\begin{itemize}
\item \emph{Mobile users}: We incorporate all the users that request cellular services within the service range of MBS, $i.e.$,
mobile users. To characterize the spatial distribution of mobile users is the first step to analyze the connections
among mobile users, MBS, and moving vehicles. The Poisson point process (PPP), which is a classical and reliable method to
model spatial distribution, has been widely used in the analysis of cellular networks. Many existing prediction methods, such
as Markov based mobility prediction algorithm \cite{zhao2018mobility}, cannot be directly applied in our scenario since it is
difficult to get
specific mobility traces of mobile users in the two-dimensional plane. Therefore, by referring to \cite{li2014throughput,
liu2016energy},\ we use the homogeneous Poisson point process (PPP) to characterize the specific
distribution of mobile users in a given range. For instance,
assuming the mean rate of PPP is $\lambda_{u}$, the probability that there are
$N_{u}$ users in the cell is obtained as $\frac{\lambda_{u}^{N_{u}}}{N_{u}%
!}e^{-\lambda_{u}}$.

\item \emph{Caching vehicles}: We assume the contact times between a user and
the vehicles caching the requested data follow the exponential distribution
with mean rate $\lambda_{v}$ \cite{vigneri2017per}\cite{karagiannis2010power}.

\item \emph{Data catalogue}: We denote the requested mobile data catalogue by a set of
file fragments with same size, ${i.e.}$, $\mathit{\digamma}$%
=$\{F_{1},F_{2},...F_{N_{f}}\}$. It should be noted that the assumption is
reasonable since files can be divided into multiple fragments with same size
\cite{vigneri2017per, golrezaei2013femtocaching}. Besides, the approach that
caches the complete file or none of it may not achieve the optimal
network performance \cite{vigneri2017per}. All file fragments have
different requested probabilities, denoted by the vector $\mathbf{p}%
=\{p_{1},...p_{N_{f}}\}$. We apply the
widely used Zipf distribution model to compute the popularity of fragments,
\textit{i.e.}, the request probability of the fragment $F_{j}$ is%
\begin{equation}
p_{j}=\frac{1}{\left(  \sum_{j=1}^{_{N_{f}}}1/j^{\phi}\right)  j^{\phi}},
\end{equation}
where $\phi$ is the Zipf exponent and its specification can be referred to
\cite{breslau1999web, gabry2016energy}. Let $q_{j}$, $j\in\left\{
1,2,...N_{f}\right\}  $ denote the probability that fragment $j$ is cached in
vehicles. The caching probability vector $\textbf{q}=\{q_{1},q_{2}%
,...q_{N_{f}}\}$ then\ can be\ formed.
\end{itemize}

%
%\[%
%\begin{array}
%[c]{cc}%
%\digamma & \text{Set of }F\text{ files}\\
%B & \text{File size}\\
%W & \text{System bandwidth}\\
%q_{j} & \text{Cache pribability of }F_{j}\\
%p_{j} & \text{Request probability of }F_{j}\\
%\phi(\sigma) & \text{Zipf exponents}\\
%N_{f} & \text{Number of files }\\
%N_{v} & \text{Number of vehicles}\\
%N_{u} & \text{Number of users}\\
%s & \text{ Coverage of MBS}\\
%\rho_{v} & \text{Density of vehicles}\\
%\mu & \text{Mean service rate of cellular downlink}\\
%\nu & \text{Mean service rate of caching vehicle}\\
%\lambda & \text{Mean arrival rate of user's requests}\\
%\eta &
%\begin{array}
%[c]{c}%
%\text{Mean inter-meeting rate with vehicles storing }\\
%\text{the requested file (exponentionally distributed)}%
%\end{array}
%\\
%\omega & \text{Impatient time of requests (exponentionally distributed)}\\
%s & \text{Coverage area of MBS}\\
%\rho_{v} & \text{Vehicle density}\\
%\sigma^{2} & \text{noise power}\\
%\epsilon & \\
%\gamma &
%\end{array}
%\]

\subsection{Communication model}\label{communicationmodel}

In this part, we describe D2D-enabled communication models in vehicular caching scenario.
Many works has been on applying D2D communications in vehicular networks, in
order to enhance the network performance of traditional ad-hoc vehicular
networks \cite{ren2015power}, \cite{liang2017resource}.
Besides, V2P (Vehicle-to-Pedestrian) communications are considered as a
similar communication mode with V2V excepted the limited power consumption of
pedestrian user according to \cite{yang2017high}. We denote by $h_{i,j}$ as
the channel gain from the transmitter $i$ to the receiver $j$, and assume a case that
caching vehicles reuse the downlink channel. According to \cite{zhang2015interference, liu2014will}, the instantaneous SINR
(signal-to-interference-plus-noise ratio) for the user $k$ served by MBS is%

\begin{equation}\label{SINR_B}
\gamma_{k}^{B}=\frac{P_{B}^{t}h_{B,k}}{\sigma^{2}+\sum_{n=1}^{N_{v}%
}\varepsilon_{n,k}P_{v,k}^{t}h_{n,k}},
\end{equation}
where $P_{B}^{t}$\ is the transmit power of MBS, and $h_{B,k}$\ denotes the
channel gain of the wireless propagation channel between MBS and the user $k$. The part of $\sum_{n=1}^{N_{v}%
}\varepsilon_{n,k}P_{v,k}^{t}h_{n,k}$ shows the interference from the vehicles
reusing the same channel with the downlink of MBS. $N_{v}$ is the number
of vehicles. $\varepsilon_{n,k}=1$ denotes that the vehicle $n$ reuses the link of
the user $k$. Similarly, the SINR for the communication link between the caching vehicle and
the user $k$ is
\begin{equation}\label{SINR_V}
\gamma_{k}^{v}=\frac{P_{v,k}^{t}h_{v,k}}{\sigma^{2}+P_{B,k}^{t}h_{B,k}%
+\sum_{n=1,n\neq v}^{N_{v}}\varepsilon_{n,m}P_{v,k}^{t}h_{n,k}},
\end{equation}
where $P_{v,k}^{t}$ is the transmit power of the caching vehicle, and
$P_{B,k}^{t}h_{B,k}$ is the interference at user $k$ caused by the MBS. Similarly with (\ref{SINR_B}),
the part of $\sum_{n=1,n\neq v}^{N_{v}}\varepsilon_{n,k}P_{v,k}^{t}h_{n,k}$ shows the interference from the vehicles reusing the same channel with the transmitting vehicle.

Following \cite{ren2015power} and \cite{cheng2015d2d}, we only consider the
large-scale fading phenomenon when computing $\gamma_{k}^{v}$. Note that, the
methods of resource allocation and channel reuse to alleviate the interference
of D2D-enabled vehicular communications have been widely studied, so that we assume
the ideal interference management.

For V2V communications, the unit disk model is used to grossly capture the fact that communication performance
depends on the transmission range of two vehicles. According to this model, the communication can directly take
place when the Euclidean distance of a pair of vehicles is within the valid communication range \cite{chen2017throughput}.

\subsection{Energy Consumption Model}\label{energymodel}

Constructing green wireless networks is a long-term target not only for traditional mobile communication networks but for next
generation vehicular networks, which makes the energy efficiency in vehicular communications concerned, especially with the popularization of electric vehicles \cite{zhang2018real}.
%In this part, to examine the influence of the
%newly proposed vehicular caching scheme on the network energy consumption, we develop
%a series of energy consumption models.
It should be noted that, using the caching
scheme may reduce the transport energy consumption from MBS at the expense of
the increase of caching energy. Therefore, our goal is to find the optimal
caching solution by taking account of the increased throughput and energy
consumption. To develop the energy consumption model in vehicular caching scheme, we only consider the
energy consumption that can be impacted by the caching
policy, \textit{i.e.}, transport energy consumed by MBS, transport energy and caching
energy consumed by caching vehicles by simplifying the energy consumption model in traditional cellular networks\cite{arnold2010power}.
%In traditional cellular networks where mobile users are served
%by MBS, the energy consumption is mainly caused by the data transport from MBS.
%In this case, {\color{red}{the energy consumption mainly includes several
%components}}, \textit{i.e.}, power amplifier, signal processing, A/D converter,
%antenna, feeder, power supply, battery backup, and cooling according to \cite{arnold2010power}.
%For simplicity, we only consider the energy consumption that can be impacted by the caching
%policy, \textit{i.e.}, transport energy consumed by MBS, transport energy and caching
%energy consumed by caching vehicles.

\begin{itemize}
\item \emph{Energy consumption for MBS}. Let $S$ denote the coverage area of a MBS. To denote the transport
energy consumed by MBS, we apply the linear energy consumption model \cite{gabry2016energy, choi2012network} as
\begin{equation}
P_{m}(t)=R_{m}(t)\omega_{t}^{m},
\end{equation}
where $\omega_{t}^{m}=0.5 * 10^{-8}$ J/bit denotes the rate of energy consumption for the
transmissions from MBS\cite{gabry2016energy}, and $R_{m}(t)$ is the total data served by MBS in time
slot $t$.

\item \emph{Energy consumption for caching vehicles}. The energy consumped by
caching vehicles consists of two parts, \textit{i.e.}, transport energy and
caching energy. Considering the backhaul transmissions from MBS to caching vehicles, we
show the energy consumption for caching vehicles as%
\begin{equation}
P_{cv}(t)=P_{t}^{v}(t)+P_{ca}(t)+P_{bh}(t),
\end{equation}
where $P_{t}^{v}$, $P_{ca},$ and $P_{bh}$ are transport energy, caching energy
and backhaul energy, respectively \cite{liu2016energy, gabry2016energy,
choi2012network}. Specifically, $P_{t}^{v}$ is a function of the transmit
power of caching vehicles, shown as
\begin{equation}\label{enegymodelequarity13}
P_{t}^{v}(t)=\zeta_{v}P_{v,k}^{t}(t),
\end{equation}
where $\zeta_{v}$\ is a simplified impact parameter for power amplifier
cooling and power supply. The energy-proportional model is used to represent
the caching energy, shown as
\begin{equation}
P_{ca}(t)=R_{v}(t)\omega_{c},
\end{equation}
where $\omega_{c}=6.25*10^{-12}$ W/bit is the caching factor for high-speed SSD devices \cite{liu2016energy}. Also, $P_{bh}$ can be obtained as%
\begin{equation}
P_{bh}(t)=R_{v}(t)\omega_{t}^{m},
\end{equation}

\item \emph{Total Energy consumption}. Based on the analysis above, the total energy consumed in the vehicular
caching network is%
\begin{equation}
P_{tot}(t)=P_{m}(t)+P_{cv}(t).
\end{equation}

\end{itemize}

%\section{Interactions of Caching Vehicles and Mobile Users}\label{Section IV}
\subsection{Interactions of caching vehicles and mobile users}

\begin{figure}[t!]
\centering
\includegraphics[width=3.3in]{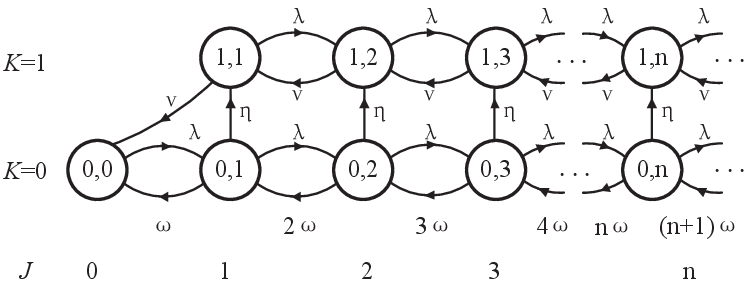}
\caption{$2$-D Markov process for the interactions of vehicles and users}%
\label{2_DMarkov}%
\end{figure}

In this part, we develop a $2$-D Markov process \cite{altman2006analysis}\cite{perel2010queues} to model the interactions between caching vehicles
and mobile users.
%which aims to obtain the service probability of users. Specifically, by referring to \cite{altman2006analysis}\cite{perel2010queues}, the detailed derivation
%about the $2$-D Markov model is shown as following.
%In our vehicular caching scheme, the service capacity of caching
%data can be enhanced because of the mobility of vehicle. However, different from mobile
%users that usually move at low speeds, vehicles often have high travel
%speeds. As such, to accurately evaluate the network availability of mobile
%users, we develop a $2$-D Markov process from the perspective of mobile users
In Fig. \ref{2_DMarkov}
, let $J$ denote the set of requests
in the user request queue and $K$ denote the network connection condition. Let
$P_{kj}=P\{K=k,J=j\}$ denote the stationary probability that request $j$ is
served by network connection $k$ ($k=0$ implies the request is served by the
cellular network, while $k=1$ denotes the D2D vehicle-user connection is working).
Since our purpose is to alleviate the network burden, the network states
in Fig. \ref{2_DMarkov}
cannot be shifted from \emph{K=1} to \emph{K=0}. This is because
that the state transition from \emph{K=0} to \emph{K=1} takes place when mobile users sense the
performance degradation of cellular networks. In this case, the communications with caching
vehicles will be the default setting in mobile users and the cellular networks assist to
provide services only when the vehicles that cache the requested data do not come or cannot develop connections
with mobile users within the users' tolerance time.

The state balance equations of the Markov process are given:%
\begin{equation}
K=1:\left\{
\begin{array}
[c]{c}%
J=1:(\lambda+\nu)P_{11}=\nu P_{12}+\xi P_{01}\\
J\geq2:(\lambda+\nu)P_{1n}=\lambda P_{1,n-1}+\nu P_{1,n+1}\\
+{\xi}P_{0n},
\end{array}
\right.  \label{BalanceEqu1}%
\end{equation}%
\begin{equation}
K=0:\left\{
\begin{array}
[c]{c}%
J=0:\lambda P_{00}=\nu P_{11}+\omega P_{01}\\
J\geq1:(\lambda+n\omega+\xi)P_{0n}=\lambda P_{0,n-1}\\
+(n+1)\omega P_{0,n+1},
\end{array}
\right.  \label{BalanceEqu2}%
\end{equation}
where $\lambda$ is the mean arrival rate of users' requests following Poisson
distribution, and $\nu$ is the mean service rate of caching vehicles. The
inter-meeting times between a user and the vehicles caching the requested
files follow exponential distribution with average rate $\xi$. For different users,
the inter-meeting times are independent and identically distributed random variables. Besides, the
tolerant time of users also follows exponential distribution with mean rate
$\omega$.

We further define the Probability Generating Functions (PGFs) as
$G_{0}(z)=\sum_{n=0}^{\infty}P_{0n}z^{n}$, and $G_{1}(z)=\sum_{n=1}^{\infty
}P_{1n}z^{n}$. Specifically, multiplying the two sides of both (\ref{BalanceEqu1}) and (\ref{BalanceEqu2}) by $z^{n}$ and summing over them, respectively, we obtain%
\begin{equation}
G_{0}^{\prime}(z)=\frac{(\lambda z-\lambda-\xi)G_{0}(z)+\xi P_{00}+vP_{11}%
}{\omega z-\omega}. \label{EquG0}%
\end{equation}
\begin{equation}
G_{1}(z)=\frac{\xi zG_{0}(z)-z\nu P_{11}-\xi zP_{00}}{\lambda z+vz-v-\lambda
z^{2}}, \label{EquG1}%
\end{equation}

By omitting the detailed derivation, we obtain $E[L_{0}]=G_{0}^{\prime
}(1)=\frac{\lambda\nu-\lambda^{2}}{\nu\omega+\upsilon\xi-\lambda\omega}$ and
$E[L_{1}]=G_{1}^{\prime}(1)=\frac{\xi\lambda}{\nu\omega+\upsilon\xi
-\lambda\omega}$, which are the number of requests served by cellular links
and vehicle-user links, respectively. Therefore, the probability that
a user is served by caching vehicles can be obtained as%

\begin{equation}
\kappa_1=\frac{E[L_{1}]}{E[L_{0}]+E[L_{1}]}.
\end{equation}

Also, the probability that serviced by MBS is
\begin{equation}
\kappa_0=\frac{E[L_{0}]}{E[L_{0}]+E[L_{1}]}.
\end{equation}
%\begin{equation}
%P_{s,v}=\frac{E[L_{1}]}{E[L_{0}]+E[L_{1}]}.
%\end{equation}
%The probability that serviced by MBS is%
%\begin{equation}
%P_{s,m}=\frac{E[L_{0}]}{E[L_{0}]+E[L_{1}]}.
%\end{equation}
The detailed calculation for $E[L_{0}]$ and $E[L_{1}]$ is shown in Appendix A.
Further, given a set of mobile users, the probability that there are $n$ users are served by caching
vehicles can be calculated as
\begin{equation}
P_{s,v}\{n\}=\binom{N_u}{n}\kappa_1^n(1-\kappa_1)^{N_u-n}.
\end{equation}
Similarly, we can obtain the probability that ${n}$ users are served by MBS as
\begin{equation}
P_{s,m}\{n\}=\binom{N_u}{n}\kappa_0^n(1-\kappa_0)^{N_u-n}.
\end{equation}

\section{Online Vehicular Caching Scheme}\label{Section V}

In this section, in order to obtain the optimal caching decisions, we first formulate vehicular caching into a
fractional optimization model towards the minimization of network energy efficiency. Using the nonlinear programming
technology and Lyapunov optimization theory, we then explore the
solution of the developed non-convex problem. Finally, we propose an online caching algorithm to achieve
the online vehicular caching.

\subsection{Problem Formulation}

Based on the analysis in \ref{communicationmodel}, the total network throughput in time slot $t$ based on Shannon's formula can be obtained as
\begin{equation}
\begin{aligned} R_{tot}(t)&=R_{m}(t)+R_{v}(t) \\ &=W\sum_{k=1}^{N_{u}}[ kP_{s,m,k}(t)\sum_{j=1}^{N_{f}}(q_{j}(t)|p_{j}(t))\log_{2}(1+\gamma_{k}^{m}(t)) \\
&+{k}P_{s,v,k}(t)\sum_{j=1}^{N_{f}}((1-q_{j}(t))|p_{j}(t))\log_{2}(1+\gamma_{k}^{v}(t))], \end{aligned}
\end{equation}where $R_{m}(t)$ and $R_{v}(t)$ are the throughput served by MBS and caching vehicles, respectively. Besides, $P_{s,m,k}(t)$ and $P_{s,v,k}(t)$ denote $P_{s,v}\{k\}$ and $P_{s,v}\{k\}$ at time slot t, respectively.
$q_j(t)$ denotes the caching probability of fragment $j$ at time slot $t$.
We adopt the same system bandwidth for both vehicle-user links and cellular links. This is typical in existing researches related with the coexistence of two kinds of
links, such as in D2D-enabled vehicular communications \cite{liang2017resource}, LTE-based V2X communications \cite{3rdgeneration2016project}, and D2D relay networks \cite{zhang2018social}.

To obtain the optimal caching policy, we focus on the energy efficiency of
networks and formulate a fractional optimization problem. From a long-term
perspective, the optimization model of network energy efficiency is

%\begin{align}
%&  \min\eta_{EE}=\underset{K\rightarrow\infty}{\lim}\frac{\frac{1}{K}%
%\sum_{t=0}^{K-1}P_{tot}\{\mathbf{q}(t)\}}{\frac{1}{K}\sum_{t=0}^{K-1}%
%R_{tot}\{\mathbf{q}(t)\}}=\frac{\overline{P}tot}{\overline{R}_{tot}%
%}\label{OptModel}\\
%s.t.\quad &
%\begin{cases}
%\label{OptModelCase} $C1$: Q_{n}(t)\text{ are mean rate stable, }\forall n\in\{1,...,N_{u}\},\\
%$C2$: \sum_{j=1}^{N_{j}}q_{j}(t)N_{v}B\leq S_{cv}N_{v},\\
%$C3$: 0\leq q_{j}(t)\leq1\text{, }\forall j\in\{1,...,N_{f}},\\
%\end{cases}
%\end{align}

\begin{align}\label{EEOptimizationModel}
\min\eta_{EE}  &  =\underset{t\rightarrow\infty}{\lim}\frac{\frac{1}{t}%
\sum_{\tau=0}^{t-1}P_{tot}(\tau)}{\frac{1}{t}\sum_{\tau=0}^{t-1}R_{tot}%
(\tau)}=\frac{\overline{P}_{tot}}{\overline{R}_{tot}}\\
\text{s.t. C1}  &  \text{: }\overline{D}_{n}\leq D_{av}\text{, }\forall
n\in\{1,...,N_{u}\},\nonumber\\
\text{C2}  &  \text{: }\sum_{j=1}^{N_{j}}q_{j}(t)B\leq S_{cv},\nonumber\\
\text{C3}  &  \text{: }0\leq q_{j}(t)\leq1\text{, }\forall j\in\{1,...,N_{f}%
\},\nonumber
\end{align}
where $\overline{D}_{n}$ is the time average\ expectation of response time that users experience,
and the constraint C1 is to guarantee the stability of user queue with data arrival.
Considering users' requirements in QoE, we assume $D_{av}=\omega$, where
$\omega$ is the mean delay tolerance of users. $S_{cv}$ is the maximum storage
capacity of each caching vehicle. The vector $\mathbf{q=\{}q_{1}(t),q_{2}%
(t),q_{3}(t),...q_{j}(t)\mathbf{\}}$ denotes the caching decisions for
mobile data. C2 is to limit the total caching capacity of caching vehicles.

\subsection{Problem Solution}

\subsubsection{Transformation}

It is obvious that the above optimization problem is nonconvex. As such, we
first transform the above fractional and tough nonconvex problem into a linear and
convex one based on the nonlinear fractional programming technology
\cite{dinkelbach1967nonlinear}.

To make the transformation, we have the following theorem.

\begin{theorem}
The problem of $\min\eta_{EE}$ can be equivalently transformed to
$\min\overline{P}tot-\eta_{EE}^{opt}(t)\overline{R}_{tot}$ \label{theorem1}
\end{theorem}

\begin{proof}
To prove Theorem \ref{theorem1}, we assume that $\mathbf{q}^{\ast}(t)$ is the
optimal caching decision vector at time slot $t$. We now prove that at any time slot, when
$\mathbf{q}^{\ast}$ is the solution of one of the two minimization problems, it
must be the solution of the other one. Specifically, we divide the
proof into two parts, \textit{i.e.}, necessity proof and sufficiency proof.

The necessity proof is to prove that $\mathbf{q}^{\ast}$ is the solution of
$\min\overline{P}tot-\eta_{EE}^{opt}\overline{R}_{tot}$ because it is the solution
of $\min\eta_{EE}$.

Specifically, since $\mathbf{q}^{\ast}$ is the optimal solution of
optimization problem (\ref{EEOptimizationModel}), we have%
\begin{equation}
\eta_{EE}^{opt}=\frac{\overline{P}_{tot}(\mathbf{q}^{\ast})}{\overline
{R}_{tot}(\mathbf{q}^{\ast})}\leq\eta_{EE}=\frac{\overline{P}_{tot}%
(\mathbf{q})}{\overline{R}_{tot}(\mathbf{q})}, \label{EE_equation18}%
\end{equation}
where $\mathbf{q}^{\ast}\in\mathbf{q}$.
We further transform (\ref{EE_equation18}) to%
\begin{equation}
\label{equality_star}\overline{P}_{tot}(\mathbf{q}^{\ast})-\eta_{EE}%
^{opt}\overline{R}_{tot}(\mathbf{q}^{\ast})=0,
\end{equation}%
\begin{equation}
\overline{P}_{tot}(\mathbf{q})-\eta_{EE}^{opt}\overline{R}_{tot}%
(\mathbf{q})\geq0,
\end{equation}%
\begin{equation}
\overline{P}_{tot}(\mathbf{q})-\eta_{EE}\overline{R}_{tot}(\mathbf{q})=0.
\end{equation}
Therefore, we can obtain the following equation.%
\begin{align}
&  \min\overline{P}_{tot}(\mathbf{q})-\eta_{EE}^{opt}\overline{R}%
_{tot}(\mathbf{q})\label{min_transformation0}\\
&  =\overline{P}_{tot}(\mathbf{q}^{\ast})-\eta_{EE}^{opt}\overline{R}%
_{tot}(\mathbf{q}^{\ast})\nonumber\\
&  =0.\nonumber
\end{align}
The proof for the necessity of Theorem \ref{theorem1} is completed.

For sufficiency proof, we aim to prove that $\mathbf{q}^{\ast}$ is the
solution of problem (\ref{EEOptimizationModel}) with the premise that
$\mathbf{q}^{\ast}$ is the solution of $\min\overline{P}tot-\eta_{EE}^{opt}\overline{R}_{tot}$.

Assuming $\mathbf{q}^{\ast}$ is the solution of (\ref{min_transformation0}),
we have
\begin{align}
&  \min\overline{P}_{tot}(\mathbf{q})-\eta_{EE}^{opt}\overline{R}%
_{tot}(\mathbf{q})\\
&  =\overline{P}_{tot}(\mathbf{q}^{\ast})-\eta_{EE}^{opt}\overline{R}%
_{tot}(\mathbf{q}^{\ast})\nonumber\\
&  =0.\nonumber
\end{align}
By rearranging above equation, we obtain
\begin{equation}
0=\overline{P}_{tot}(\mathbf{q}^{\ast})-\eta_{EE}^{opt}\overline{R}%
_{tot}(\mathbf{q}^{\ast})\leq\overline{P}_{tot}(\mathbf{q})-\eta_{EE}%
^{opt}\overline{R}_{tot}(\mathbf{q}).
\end{equation}
So that, we have%
\begin{equation}
\eta_{EE}^{opt}=\frac{\overline{P}_{tot}(\mathbf{q}^{\ast})}{\overline
{R}_{tot}(\mathbf{q}^{\ast})},
\end{equation}
and%
\begin{equation}
\eta_{EE}^{opt}\leq\frac{\overline{P}_{tot}(\mathbf{q})}{\overline{R}%
_{tot}(\mathbf{q})}.
\end{equation}
It can be seen that $\mathbf{q}^{\ast}$ is also the solution of
(\ref{EEOptimizationModel}). The sufficiency proof of Theorem \ref{theorem1}
is completed.

Therefore, the proof of Theorem \ref{theorem1} is completed.
\end{proof}

Hence, the original fractional optimization problem (\ref{EEOptimizationModel}) is
transformed to%
\begin{align}
\label{linearoptimizationmode} &  \min\overline{P}_{tot}-\eta_{EE}^{opt}(t)\overline
{R}_{tot}\\
&  \text{s.t. C1, C2, C3}.\nonumber
\end{align}
In the time slot $t$, since $\eta_{EE}^{opt}(t)$ is unknown, we relax the minimization (\ref{linearoptimizationmode}) as
\begin{align}
\label{linearoptimizationmode1} &  \min\overline{P}_{tot}-\eta_{EE}(t)\overline
{R}_{tot}\\
&  \text{s.t. C1, C2, C3}.\nonumber
\end{align}
According to \cite{dinkelbach1967nonlinear, neely2010stochastic}, this model
is a linear and convex optimization model.

\subsubsection{Virtual Queue}

Although the original fractional optimization model is transformed to a linear
and convex one, it is still difficult to directly solve it due to the existence of time-related variables. Besides,
using the traditional heuristic or iterative algorithm easily incur large
computing overhead and delay, which is intolerant in highly dynamic
communication environments. As such, we explore the application of Lyapunov
optimization theory in this study to solve the optimization problem \cite{neely2010stochastic}.
Before that, the primary question is to tackle the time-related inequality
constraint C1. To this end, the virtual queue technology is used to transform the time-related variable
into a problem of queue stability \cite{neely2013dynamic,
neely2006energy}. Specifically, for constraint C1, the virtual queue
$H_{n}(t)$ for the user $n$ is defined as%
\begin{equation}
\label{virtualqueue}H_{n}(t+1)=\max[H_{n}(t)+e_{n}(t),0],
\end{equation}
where $e_{n}(t)=D_{n}(t)-D_{av}$. Based on the defination, we give Theorem \ref{Theorem2} below.

\begin{theorem}
\label{Theorem2} The constraint C1 can be satisfied by guaranteeing that the
virtual queue is mean rate stable.
\end{theorem}

\begin{proof}
The equation (\ref{virtualqueue}) can be recast as
\begin{equation}
H_{n}(t+1)\geq H_{n}(t)+e_{n}(t).
\end{equation}
By summing over $t\in\{0,...K-1\}$ at both sides of above inequality and
rearranging terms, we have
\begin{equation}\label{equation31HnK}
H_{n}(K)\geq H_{n}(0)+\sum_{t=0}^{K-1}e_{n}(t).
\end{equation}
If we divide both sides of (\ref{equation31HnK}) by $K$, and take the value of $K$ going to
infinity, we obtain
\begin{equation}
\underset{K\rightarrow\infty}{\lim}\frac{1}{K}H_{n}(K)\geq
\underset{K\rightarrow\infty}{\lim}\frac{1}{K}H_{n}(0)+\underset{K\rightarrow
\infty}{\lim}\frac{1}{K}\sum_{t=0}^{K-1}e_{n}(t) \label{ExpectionBoth}.%
\end{equation}
Taking an expectation for (\ref{ExpectionBoth}), the equation becomes
\begin{equation}
\underset{K\rightarrow\infty}{\lim}\frac{1}{K}E[H_{n}(K)]\geq
\underset{K\rightarrow\infty}{\lim}\frac{1}{K}E[\sum_{t=0}^{K-1}%
D_{n}(t)]-D_{av}. \label{Inequality33}%
\end{equation}
According to Jeasen's theory, the following inequality can be obtained%
\begin{equation}
0\leq|E\{H_{n}(t)\}|\leq E\{|H_{n}(t)|\}.
\end{equation}
Because $H_{n}(t)$ is mean rate stable, we have%
\begin{equation}
\underset{K\rightarrow\infty}{\lim}\frac{E\{|H_{n}(K)|\}}{K}=0.
\end{equation}

Therefore, due to $H_{n}(K)\geq0$, it is obvious that $|E\{H_{n}%
(t)\}|=E\{H_{n}(t)\}$. Hence, the left side of (\ref{Inequality33}) should be
\begin{equation}
D_{av}\geq\overline{D}_{n},
\end{equation}
where $\overline{D}_{n}=\underset{K\rightarrow\infty}{\lim}\frac{1}{K}E[\sum_{t=0}^{K-1}D_{n}(t)]$.

The proof of Theorem \ref{Theorem2} is completed.
\end{proof}

\subsubsection{Lyapunov Optimization}

In this part, we aim to apply the Lyapunov optimization theory to solve the
optimization problem (\ref{linearoptimizationmode1}). Firstly, we need to define the Lyapunov function as follows.

Let $\mathbf{\Theta}(t)\overset{\bigtriangleup}{=}\mathbf{H}(t)$ denote the combined
queue backlog vector where $\mathbf{\Theta}(t)=(\theta_{1}(t),\theta_{2}(t),\theta_{3}(t),...\theta_{N}(t))$, the quadratic
polynomial of\ Lyapunov function is defined as \cite{neely2010stochastic}%
\begin{equation}
L(\mathbf{\Theta}(t))\overset{\vartriangle}{=}\frac{1}{2}\sum_{n=1}^{N_{u}%
}H_{n}(t)^{2}.%
\end{equation}
After that, the one-slot conditional Lyapunov drift can be obtained as%
\begin{equation}
\label{conditionaldrift}\Delta(\mathbf{\Theta}(t))\overset{\vartriangle
}{=}E\{L(\mathbf{\Theta}(t+1))-L(\mathbf{\Theta}(t))|\mathbf{\Theta}(t)\}.
\end{equation}
Further, we use the drift-plus-penalty to guarantee the stability of the virtual queue and solve the optimization problem.
By the drift-plus-penalty, the problem (\ref{linearoptimizationmode1}) can be solved as
\begin{equation}\label{mindrifypluspenalty}
\min\Delta(\mathbf{\Theta}(t))+VE\{P_{tot}(t)-\eta_{EE}(t)R_{tot}(t)|{\mathbf{\Theta}(t)}\}.
\end{equation}

To solve the minimization problem (\ref{mindrifypluspenalty}), we give Theorem \ref{Theorem3} as follows.

\begin{theorem}
\label{Theorem3} The bound of the drift-plus-penalty can be written as
\begin{equation}%
\begin{split}\label{drift_plus_penalty}
\Delta(\mathbf{\Theta}(t))+V  &  E\{P_{tot}(t)-\eta_{EE}(t)R_{tot}%
(t)|\mathbf{\Theta}(t)\}\leq B\\
&  +\sum_{n=1}^{N_{u}}H_{n}(t)E\{e_{n}(t)|\mathbf{\Theta}(t)\}\\
&  +VE\{P_{tot}(t)-\eta_{EE}(t)R_{tot}(t)|\mathbf{\Theta}%
(t)\},%
\end{split}
\end{equation}

where
\begin{equation}
B\geq\frac{1}{2}\sum_{n=1}^{N_{u}}E\{e_{n}(t)^{2}|\mathbf{\Theta}(t)\}.
\end{equation}
\end{theorem}

\begin{proof}
By squaring both sides of equation (\ref{virtualqueue}) and rearranging terms,
we have
\begin{equation}
\label{virtualqueue1}\frac{1}{2}[H_{n}(t+1)^{2}-H_{n}(t)^{2}]\leq\frac{1}{2}%
e_{n}(t)^{2}+H_{n}(t)e_{n}(t).
\end{equation}
%We further define the one-slot conditional Lyapunov drifty as%
%\begin{equation}
%\label{conditionaldrift}\Delta(\mathbf{\Theta}(t))\overset{\vartriangle
%}{=}E\{L(\mathbf{\Theta}(t+1))-L(\mathbf{\Theta}(t))|\mathbf{\Theta}(t)\}
%\end{equation}
Summing over $n\in\{1,...N_{u}\}$ for (\ref{virtualqueue1}) and taking a conditional expectation, we have%
\begin{equation}%
\begin{split}
\frac{1}{2}\sum_{n=1}^{N_{u}}E&\{H_{n}(t+1)^{2}-H_{n}(t)^{2}|\mathbf{\Theta}(t)\}
  \leq\\
  &\sum_{n=1}^{N_{u}}\frac{1}{2}E\{e_{n}(t)^{2}|\mathbf{\Theta}(t)\}+\sum_{n=1}^{N_{u}}H_{n}(t)E\{e_{n}(t)|\mathbf{\Theta}(t)\}.
\end{split}
\end{equation}
According to (\ref{conditionaldrift}), the above equation equals to%
\begin{equation}
\label{deltatheta}\Delta(\mathbf{\Theta}(t))\leq\frac{1}{2}\sum_{j=1}%
^{J}E\{e_{j}(t)^{2}|\mathbf{\Theta}(t)\}+\sum_{j=1}^{J}H_{j}(t)E\{e_{j}%
(t)|\mathbf{\Theta}(t)\}.
\end{equation}
Adding $VE\{P_{tot}(t)-\eta_{EE}(t)R_{tot}(t)|\mathbf{\Theta}(t)\}$ on both sides of (\ref{deltatheta}), it becomes
\begin{equation}%
\begin{split}
\Delta(\mathbf{\Theta}(t))+VE\{  &  P_{tot}(t)-\eta_{EE}(t)R_{tot}%
(t)|\mathbf{\Theta}(t)\}\leq B\\
&  +VE\{P_{tot}(t)-\eta_{EE}(t)R_{tot}(t)|\mathbf{\Theta}(t)\}\\
&  +\sum_{j=1}^{J}H_{j}(t)E\{e_{j}(t)|\mathbf{\Theta}%
(t)\}\label{right_side_Lypunovdelta}.%
\end{split}
\end{equation}
Therefore, the equation (\ref{drift_plus_penalty}) can be proved, where
\begin{equation}
B\geq\frac{1}{2}\sum_{j=1}^{J}E\{e_{j}(t)^{2}|\mathbf{\Theta}(t)\}.
\end{equation}

The proof of Theorem \ref{Theorem3} is completed.
\end{proof}

Based on above analysis, the optimization problem of
(\ref{linearoptimizationmode1}) can be solved by minimizing the right-hand-side of
inequality (\ref{right_side_Lypunovdelta}) rather than directly minimizing
$\Delta(\mathbf{\Theta}(t))+VE\{P_{tot}(t)-\eta_{EE}(t)R_{tot}%
(t)|\mathbf{\Theta}(t)\}$.

\subsubsection{Performance of Lyapunov Optimization}

In this part, we further conduct a theoretical performance analysis about the
Lyapunov optimization based the solution above.

\textbf{Boundedness Assumptions}

Before the anslysis, we first give general assumptions as follows.
\begin{equation}
\label{Assumption1}E\{R_{tot}(t)\}\in\lbrack R_{\min},R_{\max}],
\end{equation}
\begin{equation}
\label{Assumption2}E\{P_{tot}(t)\}\in\lbrack P_{\min},P_{\max}],
\end{equation}
\begin{equation}
\label{Assumption3}0\leq q(t)\leq1,
\end{equation}
\begin{equation}
\label{Assumption4}E\{P_{tot}^{\ast}(t)\}\leq E\{R_{tot}^{\ast}(t)\}(\eta
_{EE}^{opt}+\delta^{2}),
\end{equation}
\begin{equation}
\label{Assumption5}E\{e_{n}(t)|\mathbf{\Theta}(t)\}=E\{e_{n}(t)\}\leq
\delta^{2}.
\end{equation}
Apparently, the assumptions (\ref{Assumption1})-(\ref{Assumption3}) are
reasonable since the three variables must be limited in specified
ranges. The assumption (\ref{Assumption4}) is reasonable according to
(\ref{equality_star}). For the assumption (\ref{Assumption5}), since $e_{n}(t)$ may take positive or negative value, we thus assume that its
expectation is less than a finite constant $\delta^{2}$.

\textbf{Performance Analysis}

In Theorem \ref{Theorem3}, we successfully prove that the problem
(\ref{linearoptimizationmode1}) is equivalent to minimizing the right-hand-side of
(\ref{drift_plus_penalty}) based on the Lyapunov optimization theory. As such, we further
explore the performance of the right-hand-side minimization. Specifically, assuming
that the optimal caching decision $\mathbf{q}^{\ast}$ is obtained by the
right-hand-side minimization, there are several properties we aim to discuss.

It should be noted that the virtual queue is mean rate stable, since the
inequality (\ref{drift_plus_penalty}) satisfies the basic form of
drift-plus-penalty in\cite{neely2010stochastic}. Therefore, the constraint C1
is satisfied according to Theorem \ref{Theorem2}.

By substituting the boundedness assumptions into (\ref{drift_plus_penalty})
and taking $\delta\rightarrow0$, we have
\begin{equation}%
\begin{split}
\label{inequality53}\Delta(\mathbf{\Theta}(t))+VE\{  &  P_{tot}(t)-\eta
_{EE}(t)R_{tot}(t)|\mathbf{\Theta}(t)\}\leq B\\
&  +VE\{P_{tot}^{\ast}(t)-\eta_{EE}(t)R_{tot}^{\ast}(t)|\mathbf{\Theta}(t)\}.
\end{split}
\end{equation}
Due to $E[E[X_2|X_1]]=E[X_2]$, we then take an expectation on both sides of inequality
(\ref{inequality53}), it thus becomes
\begin{equation}%
\begin{split}
&  E\{L(\mathbf{\Theta}(t+1))\}-E\{L(\mathbf{\Theta}(t))\}+VE\{P_{tot}(t)\\
&  -\eta_{EE}(t)R_{tot}(t)\}\leq B+VE\{P_{tot}^{\ast}(t)-\eta_{EE}%
(t)R_{tot}^{\ast}(t)\}.
\end{split}
\end{equation}
Summing over from $t\in\{0,...K-1\}$, we get
\begin{equation}%
\begin{split}
\label{inequaltiy54}E\{L(\mathbf{\Theta}(K))\}-  &  E\{L(\mathbf{\Theta
}(0))\}+V\sum_{t=0}^{K-1}E\{P_{tot}(t)\\
&  -\eta_{EE}(t)R_{tot}(t)\}\leq KB+KVE\{P_{tot}^{\ast}(t)\}\\
&  -VE\{R_{tot}^{\ast}(t)\}\sum_{t=0}^{K-1}\eta_{EE}(t).
\end{split}
\end{equation}
Plugging assumption (\ref{Assumption4}) into (\ref{inequaltiy54}), we then
divide (\ref{inequaltiy54}) by $VK$ and take $K\rightarrow\infty$, it becomes%
\begin{equation}%
\begin{split}
\label{inequality55}&\underset{K\rightarrow\infty}{\lim}\frac{1}{K}%
\sum_{t=0}^{K-1}E\{P_{tot}(t)-\eta_{EE}(t)R_{tot}(t)\}\leq\frac{B}{V}\\
&  +\eta_{EE}^{opt}E\{R_{tot}^{\ast}(t)\}-E\{R_{tot}^{\ast}%
(t)\}\underset{K\rightarrow\infty}{\lim}\frac{1}{K}E\{\sum_{t=0}^{K-1}%
\eta_{EE}(t)\},
\end{split}
\end{equation}
where
\begin{equation}
\underset{K\rightarrow\infty}{\lim}\frac{1}{K}
\sum_{t=0}^{K-1}E\{L(\mathbf{\Theta}(K))\}-E\{L(\mathbf{\Theta}(0))\}=0.
\end{equation}
According to the Lebesgue dominated convergence theorem, we obtain
\begin{equation}
\underset{K\rightarrow\infty}{\lim}\frac{1}{K}E\{\sum_{t=0}^{K-1}\eta
_{EE}(t)\}=\eta_{EE},
\end{equation}
where we assume $\underset{K\rightarrow\infty}{\lim}\frac{1}{K}\sum_{t=0}^{K-1}\eta
_{EE}(t)$ is convergent.
Rearranging (\ref{inequality55}) we have%
\begin{equation}
\frac{B}{V}+\eta_{EE}^{opt}E\{R_{tot}^{\ast}(t)\}\geq E\{R_{tot}^{\ast
}(t)\}\eta_{EE}.%
\end{equation}

Finally, we obtain the upper bound of $\eta_{EE}$ as%
\begin{equation}\label{equation59}
\eta_{EE}\leq\frac{B}{VE\{R_{tot}^{\ast}(t)\}}+\eta_{EE}^{opt}.%
\end{equation}
From (\ref{equation59}) we can see, $\eta_{EE}$ will approach to $\eta_{EE}^{opt}$
by increasing the parameter $V$.

From the analysis above, the original optimization problem (\ref{EEOptimizationModel}) now can be solved by minimizing the right-hand-side of (\ref{drift_plus_penalty}). The minimization of the right-hand-side of (\ref{drift_plus_penalty}) equals to
\begin{align}\label{FinalOptimization}
\min\sum_{n=1}^{N_{u}}H_{n}(t)D_{n}(t)+&VE\{P_{tot}(t)-\eta_{EE}(t)R_{tot}(t)|\mathbf{\Theta}(t)\}\nonumber\\
& \text{s.t. C2, C3}.
\end{align}

The new optimization model (\ref{FinalOptimization}) obtains a trade-off between the minimization of energy efficiency and the stability of the virtual queue. In this case, the large $V$ will achieve a good energy efficiency performance with the cost of the performance degradation of virtual queue stability. Therefore, to select an appropriate value for $V$ that balances the performance of energy efficiency and virtual queue stability is critical.

\subsection{Online Caching}

Based on the conclusion above, we develop an online caching scheme in Algorithm \ref{Algorithm 1}. Specifically, at each time
slot, the cache decision for next time slot is determined in MBS by optimizing the network energy efficiency.
Once the caching decision is received by a caching vehicle, it will
compare the decision from MBS with their own caching data and in case of similarity, it does not need to update the current
cache, otherwise, it
should renew the cache based on the received decision from MBS. There are two methods to update caching data based
on CCNs, $i.e.,$
caching vehicles can get the updating data based on the naming formation from nearby vehicles or from MBS.
In the next time slot, mobile users will enjoy the services under the
cooperation of caching vehicles and MBS. Algorithm \ref{Algorithm 1} shows the detailed decision-making process. Firstly, using
the variables $H_n(t)$ and $\eta_{EE}(t)$, the decision $\mathbf{q}^{\ast}(t+1)$ can be calculated by (\ref{FinalOptimization})
and then be broadcasted to caching vehicles. Secondly, $H_n(t)$ and $\eta_{EE}(t+1)$ are updated according to their renewal
processes, respectively. Finally, MBS switches to next time slot and prepare to a new process.

\begin{algorithm}[!t]
  \caption{Online Caching Algorithm.}
  \renewcommand{\algorithmicrequire}{\textbf{Input:}}
  \renewcommand{\algorithmicensure}{\textbf{Output:}}
  \label{Algorithm 1}
  \begin{algorithmic}[1]
    \Require
      $H_n(t)$, $\eta_{EE}(t)$;%, $\mathbf{q}^{\ast}(t)$;
    \Ensure
      $H_n(t+1)$, $\eta_{EE}(t+1)$, $\mathbf{q}^{\ast}(t+1)$;
    \State At the time slot $t$, get the current variables and obtain the optimal caching decision $\mathbf{q}^{\ast}(t+1)$ by solving
    (\ref{FinalOptimization});
    \label{step1}
    \State Update $H_n(t+1)$ according to (\ref{virtualqueue});
    \label{step2}
    \State Update $\eta_{EE}(t+1)$ according to;
        \begin{align*}
          {\eta_{EE}(t)} &  =\frac{\sum_{t=0}^{K-1}P_{tot}(\mathbf{q}^{\ast}(t+1))}{\sum_{t\rightarrow0}^{K-1}R_{tot}(\mathbf{q}^{\ast}(t+1))},
        \end{align*}
    \label{step3}
    \State $t=t+1$;
    \label{step4}
    %\Return $E_n$;
  \end{algorithmic}
\end{algorithm}
\begin{remark}
In the online caching scheme, the decision at next time slot is made by the current network state. However, due
to the highly dynamic network environments and the uncertainty of mobile users, the hit ratio of caching
vehicles may be influenced with different slot lengths. In Section \ref{Section VI}, we do not
evaluate the impact of different slot lengths on the performance of the online vehicular caching. This is because that the
problem above can be easily solved by existing predication algorithms, such as deep learning
\cite{jia2016fusing}. The prediction based caching decision is also our next work.
\end{remark}

\section{Performance Evaluation}\label{Section VI}

In this section, the performance of our proposal is evaluated by extensive simulations on Matlab.

\subsection{Simulation Settings}

We consider a simple but practical urban simulation scenario. In the cell
served by a MBS with a disk of radius $R=350$ meters, we select a road segment that has four-lane bidirectional traffic flow.
The system bandwidth is set as 10 MHz, which is a typical setting in D2D/LTE related researches,
as shown in \cite{liang2017resource, 3rdgeneration2016project, zhang2018social}.
The performance under the joint service of caching vehicles
and MBS is simulated. Specifically, to simulate the behavior of vehicles, we use the real-world mobility traces of taxi
cabs, collected from the GPS coordinates of approximately 320 taxis over 30
days in Rome, Italy \cite{roma-taxi-20140717}. According to the data set, we
extrapolate the statistics of vehicle distribution, which then is
assumed as the Poisson distribution in our simulation. We model the spatial
distribution of mobile users as homogeneous Poisson point process
(PPP)\cite{liu2016energy}. The probability that there are $n$ mobile users in
a given region thus can be obtained.
%The maximum communication distance between
%vehicles and users is assumed as $300$m, which is same with the vehicular communications \cite{luan2015social}.
%The transmit powers $P_{v,k}^t$ and $P_{B}^t$ are set as $23$dBm and $46$dBm, respectively. The value of $\zeta_{v}$
%is assumed as $15.13$. The noise power $\sigma^{2}$ is $-110$dBm.
The parameter $V$, which is used to control the trade-off between queue stability and energy efficiency, has been
widely studied in existing literature \cite{neely2010stochastic, neely2013dynamic, neely2006energy}. We therefore opt for
$V=50$ that guarantees the convergence according to their results, in order to fully focus on evaluating the performance with
the variations of request arrival rate, cache proportion and cache capacity.
Other parameters in our simulations are shown in TABLE \ref{TABLE1}.

\begin{table}[!t] %��'һ������environment��������λ����h,here ��
\caption{Simulation Parameters}
\centering
\label{TABLE1}
\begin{tabular}{p{1.4cm}|p{1cm}|p{1.4cm}|p{1cm}} %������ÿһ�еĿ��ȣ�ǿ��ת����
\hline
\hline
Parameters & Value & Parameters & Value \\ % ��&���ָ���Ԫ�������� \\��ʾ������һ��

\hline %��һ�����ߣ������ľͶ���һ���ˣ�����һ����4 ������
       \emph{$P_{v,k}^t$} &  $23$dBm   & \emph{$\sigma^{2}$}   &$-110$dBm  \\
\hline
       \emph{$P_{B}^t$} &  $46$dBm      & \emph{$W$} & $10$MHz \\
\hline
       \emph{$\zeta_{v}$} & $15.13$   & \emph{$\phi$} & $0.7$ \\
\hline
       \emph{$B$} &  $1000$    & \emph{$F_{n}$ } & $10$Mb \\
%\hline
%       \emph{$W$}&  $10$MHz   & \emph{ }      &  $ $ \\
%\hline
%       \emph{$\phi$} &$0.7$    & \emph{ }      &   $ $ \\
\hline
\hline
\end{tabular}
\end{table}

\subsection{Simulation Results}

Fig. \ref{eevslmada}
shows the relationship between energy efficiency and request arrival rate. In this experiment, we assume that the
normalized caching capacity is 0.01, $i.e.$, each caching vehicle caches at most $1\%$ of total data due to the storage limitation. Besides,
the cache proportion is set to $0.5$, which means that half of vehicles act as cache carriers. Fig. \ref{eevslmada}
compares the energy efficiency ($\eta_{EE}$) of online
vehicular caching (Online VC) with that of offline vehicular caching (Offline VC) and no vehicular caching (No VC). The lower $\eta_{EE}$ means that the lower energy is
consumed by data transmissions, which demonstrates the better performance in energy consumption. With different $\lambda$, the results of No VC remain at a steady level. The
Offline VC means that the cache decision is updated at a
relatively long time interval, $e.g.$, one day, as illustrated in \cite{Vigneri2017Quality}. This approach may
reduce the backhaul energy consumption but the real-time and optimal hit ratio cannot be obtained. For Offline VC, the energy
efficiency increases gradually when $\lambda$ is small. When $\lambda\geq0.6$, the $\eta_{EE}$ of Offline VC also
remains at a steady level. For Online VC, the $\eta_{EE}$ is apparently increased with the increase of $\lambda$
though the increment gradually decreases.
The large $\lambda$ also makes the $\eta_{EE}$ of Online VC steady gradually. This is because that too many requests of mobile
users exceed the service capacity
of caching vehicles. In this case, much more users will be served by MBS. Therefore, the strength of energy
performance incurred by caching vehicles is gradually reduced. Averagely, the energy efficiency of Online VC improves about
$35.8$\%
compared with No VC, and about $12.63$\% compared with Offline VC. The reason that Online VC outperforms Offline VC is that
Online VC
covers the changing needs of mobile users
by making caching decisions in a small slot time based on the newly proposed algorithm.

\begin{figure}[tb]
\centering
\includegraphics[width=3.35in]{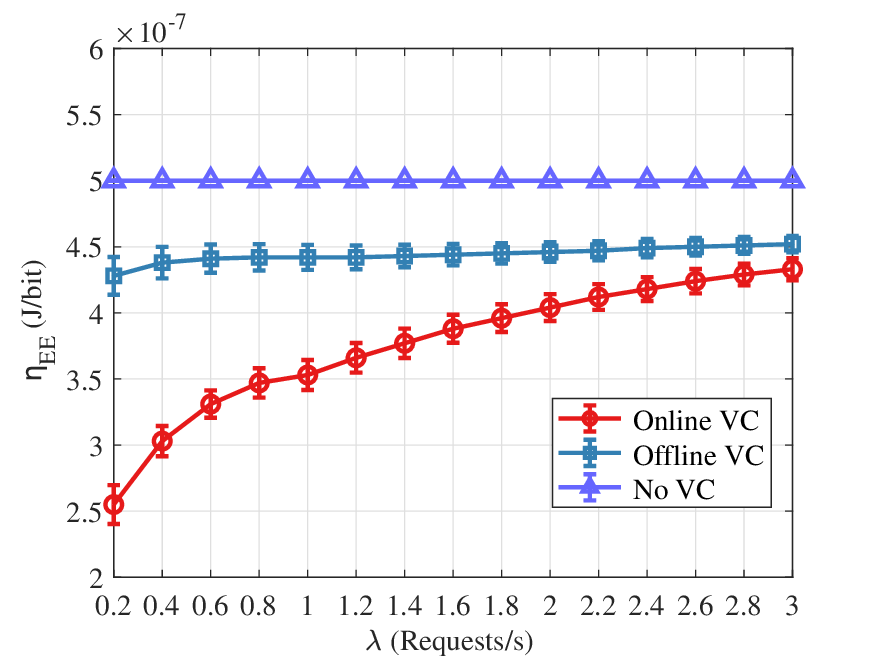}\
\caption{Comparison for the energy efficiency of three schemes}%
\label{eevslmada}%
\end{figure}

\begin{figure}[!t]
\centering
\includegraphics[width=3.5in]{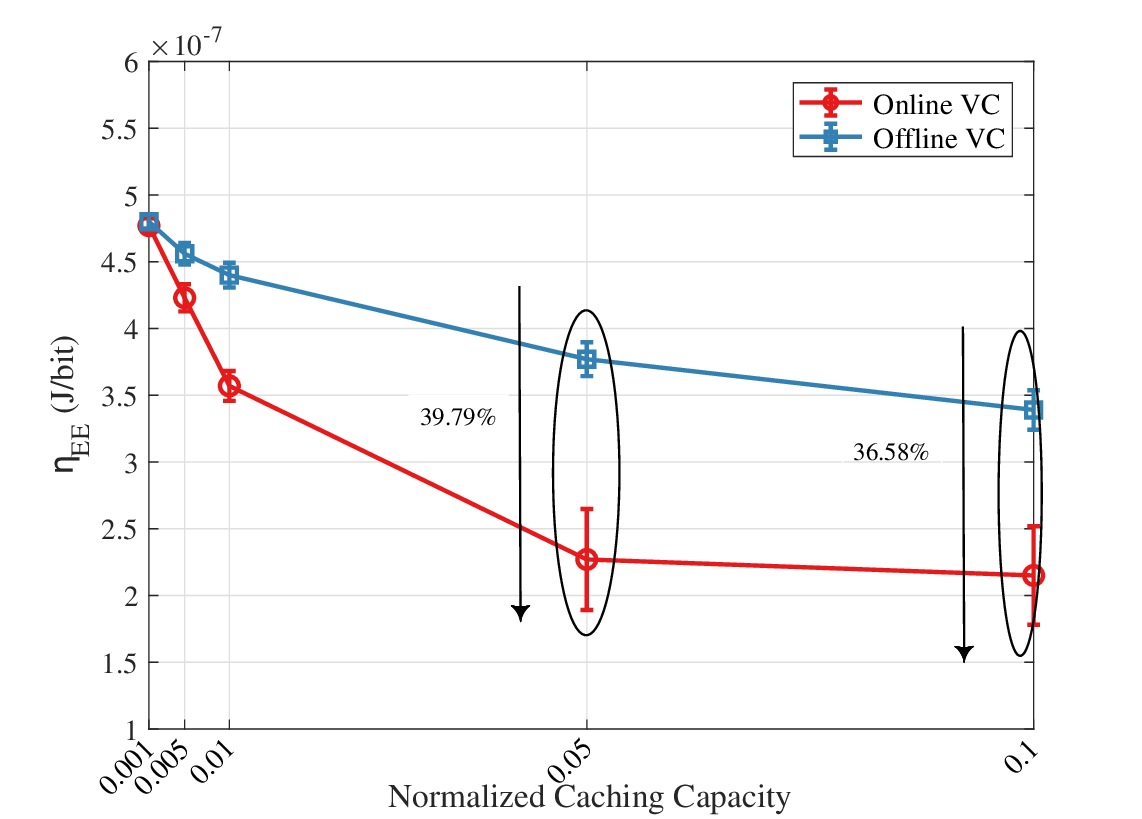}
\caption{Impact of normalized caching capacity on energy efficiency}%
\label{eevseta}%
\end{figure}

\begin{figure}[!t]
\centering
\includegraphics[width=3.4in]{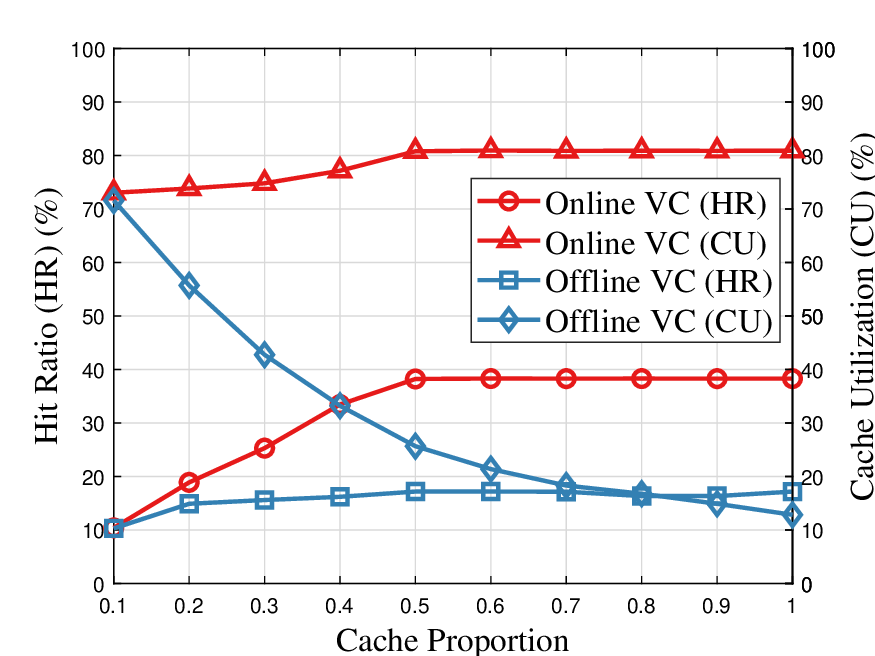}
\caption{Impact of cache proportion on hit ratio and cache utilization}%
\label{hrvscp2}%
\end{figure}

Fig. \ref{eevseta}
shows the energy efficiency performance with the variations of normalized caching capacity. In this experiment, the cache
size of each caching vehicle, referred as normalized caching capacity, is set in the range of $0.1-10\%$ of the total data catalogue. Besides, $\lambda$ is set to $1$ and the cache proportion
is set to $0.5$. The comparison of $\eta_{EE}$ between Online VC and Offline VC is conducted. In Fig. \ref{eevseta}
, it can
be seen that Online VC always outperforms Offline VC in different normalized caching capacity, which verifies the efficiency of the
newly proposed algorithm. The gap between these two schemes reaches a maximum value when the normalized cache capacity is
$0.05$, $i.e.$, Online VC improves about $37.79\%$ compared with Offline VC. Specifically, the gap between two schemes is small
when the caching capacity is small. This is because that only a small part of mobile users can be served by caching
vehicles, resulting in the minor gain. With the increase of cache capacity,
the gap between two schemes increases, which is because more users are served by caching vehicles. In this case, the advantage
of our proposal is obvious.
When the cache capacity reaches $0.05$, the gap is gradually reduced and tends to be steady. This is because that
the vehicular
cache reaches a saturated condition, $i.e.$, the cache capacity of caching vehicles is enough to provide caching services to
mobile users. In this case, the increase of cache capacity will result in extra cache energy on both schemes.

\begin{figure}[!t]
\centering
\includegraphics[width=3.4in]{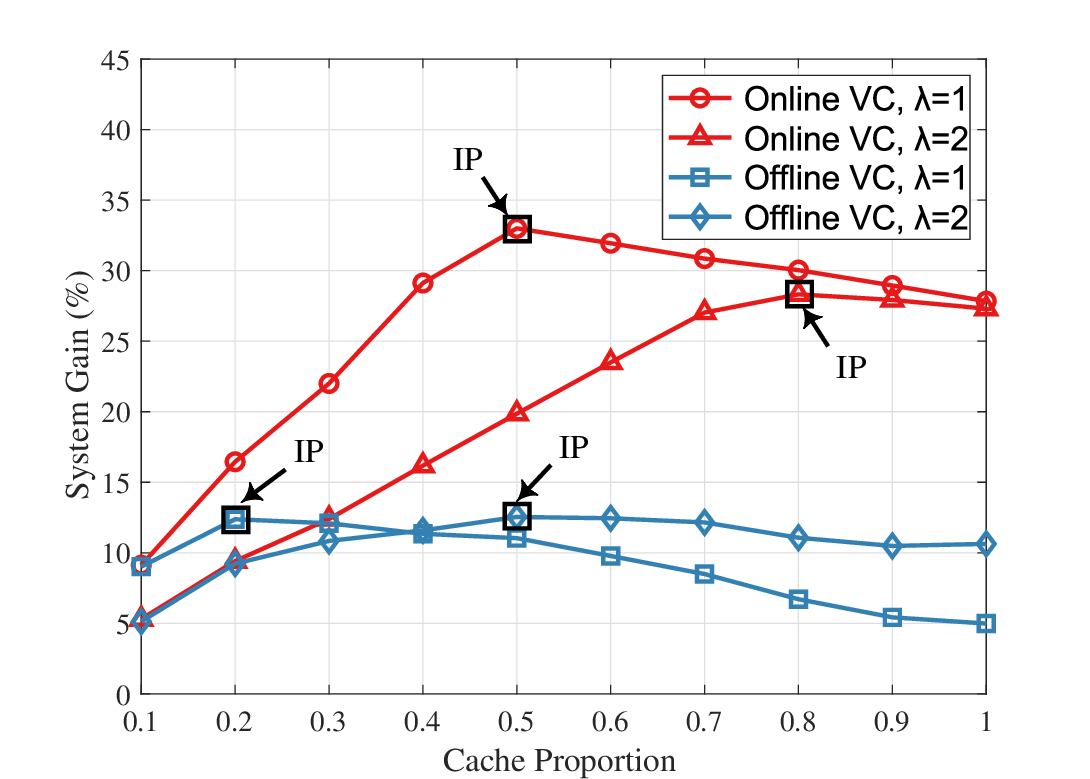}
\caption{Comparison of system gain in different scenarios}%
\label{suvscp}%
\end{figure}
%\vspace{-0.8cm}

Fig. \ref{hrvscp2}
is to explore the impact of cache proportion on hit ratio and cache utilization. In this
experiment, we still set $\lambda=1$, and assume the normalized cache capacity is $1\%$. The hit ratio refers to
the proportion of users served by caching vehicles when they send requests to vehicles. If the requests cannot be responded by caching vehicles in tolerance time, mobile users will switch to communicate with MBS. As such, the higher hit ratio demonstrates that the vehicular caching has more powerful service capability. The cache utilization denotes the utilization of mobile data
cached by vehicles, \emph{i.e.}, the proportion that mobile data in caching vehicles is accessed by mobile users. The ultimate
value of cache utilization ($100\%$) means that the mobile data updated in caching vehicles one times can serve mobile users all the time. In this case, there is no need to update vehicular cache, which will significantly reduce the network overhead.
Specifically, when the cache proportion is small, there is a rising tendency in hit ratio
for both online and offline caching schemes. The Online VC firstly achieves a major increase and then remains a
steady level, which always outperforms the Offline VC. Averagely, the Online VC improves about $10.4\%$ in hit ratio
compared with the Offline VC. However, when the hit ratio reaches a steady state, to increase cache
proportion will incur extra backhaul energy and caching energy, resulting in the performance degradation of
energy efficiency. This conclusion can be a guidance to determine the cache proportion that
achieves the optimal energy efficiency. On the other hand, with the increase of cache proportion, the cache utilization
obtained by Online VC always remains a relatively stable range, $i.e.$, $[73\%,81\%]$. This stability is
achieved by the real-time optimization based on the newly proposed algorithm. However, Offline VC, which
updates the cache content at a long time interval, obtains a declining cache utilization with the increase of
cache proportion. For cache utilization, Online VC achieves at most $80.63\%$ performance improvement compared with the Offline
VC.

Fig. \ref{suvscp}
presents the system gain caused by the two caching schemes with the variations of cache
proportion. The system gain represents the extra system throughput produced by the vehicular cache, whose increase
is because that multiple times services may be performed by a caching vehicle after caching data. Actually, the reason behind
the increase of system gain is correspondence with cache utilization in Fig. \ref{hrvscp2}. In
this experiment, the normalized cache capacity is also $1\%$. We opt for two values for $\lambda$, $i.e.$,
$\lambda=1$ or $\lambda=2$, to make a comparison. It can be seen that four curves in Fig. \ref{suvscp}
have
different Inflection Points (IPs). The existence of IP is because that when the system gain reaches a maximum value,
the service of vehicular caching is saturated, which means the increase of cache proportion no longer
contributes to the improvement of system throughput. Specifically, when $\lambda=1$, the IP is reached by Online VC at the point where the
cache proportion is $0.5$, and by Offline VC when the cache proportion is $0.2$. These results show that caching vehicles play a stronger role in Online VC compared with Offline VC. Therefore, the Online VC makes the network accommodate more users compared with Offline VC, accordingly resulting in more system gain.
In Fig. \ref{suvscp}
, the maximum system gain is about $33\%$, obtained by Online VC. The similar conclusion can
also be obtained under $\lambda=2$ while the maximum system gain is about $28\%$. Besides, the figure shows
that the Online VC achieves a better system gain compared with Offline VC in a given cache proportion, which also verifies the effectiveness of our proposal.

\section{Conclusion}\label{Section VII}

This paper focuses on enabling the efficient and reliable vehicular caching in cellular networks.
Specifically, we first develop a $2$-D Markov process to model
the communications of caching vehicles and mobile users. The probability that mobile users served by caching vehicles or MBS can be calculated.
Further, The D2D communication technology is
used to evaluate the network throughput in our scenario. By incorporating the D2D communications and a series of energy
consumption models, the caching decision
problem then be formulated as a fractional optimization model, targeting on the optimization of energy efficiency. To the best
of our knowledge, this is the
first study that takes the energy efficiency as an optimization goal in vehicular caching, which is because that the energy
management is a promising and
crucial problem with the popularization of electric vehicles in the future. We then use the nonlinear fractional programming
technology and Lyapunov optimization
theory to explore the solution of the above optimization problem. The problem thus can be transformed into a linear and convex
one. Based on the solution, we develop an online caching algorithm for vehicular caching scheme. This, to the best of our
knowledge, is also the first literature to apply the
Lyapunov in the research of vehicular caching, which can provide an important reference to other related studies. By
extensive simulations based on Matlab, the performance of the online vehicular caching is evaluated in
terms of energy efficiency, hit ratio, cache utilization and system gain. The comparison results with other schemes also verify the effectiveness of our proposal.

\appendices\label{APPENDIXA}
\section{Proof for the Derivation of $E[L_{0}]$ and $E[L_{1}]$}

%The proof of derivation $E[L_{0}]$ and $E[L_{1}]$ is shown as follow.

For simplicity, we assume $M=\xi P_{00}+vP_{11}$. The equation (\ref{EquG0}) becomes%
\begin{equation}
G_{0}^{\prime}(z)+\frac{(-\lambda z+\lambda+\xi)}{\omega z-\omega}%
G_{0}(z)=\frac{M}{\omega z-\omega}. \label{EquG00}%
\end{equation}
Assuming that $f(x)=\frac{-\lambda z+\lambda+\xi}{\omega z-\omega}$, we have
\begin{equation}\label{}
%e^{\int f(x)dx}=e^{-\frac{\lambda}{\omega}z+\frac{\eta}{\omega}\ln(z-1)}=e^{-\frac{\lambda}{\omega}z}(z-1)^{\frac{\eta}{\omega}}.
  e^{\int f(x)dx}=e^{-\frac{\lambda}{\omega}z+\frac{\xi}{\omega}\ln(z-1)}=e^{-\frac{\lambda}{\omega}z}(z-1)^{\frac{\xi}{\omega}}.
\end{equation}
Dividing the terms of
(\ref{EquG00}) by $e^{-\frac{\lambda}{\omega}z}(z-1)^{\frac{\xi}{\omega}}$, we obtain
\begin{equation}
\begin{split}
G_{0}^{\prime}(z)e^{-\frac{\lambda}{\omega}z}(z-1)^{\frac{\xi}{\omega}%
}+&e^{-\frac{\lambda}{\omega}z}(z-1)^{\frac{\xi}{\omega}}\frac{(-\lambda
z+\lambda+\xi)}{\omega z-\omega}G_{0}(z)\\
&=\frac{M}{\omega z-\omega}e^{-\frac{\lambda}{\omega}z}(z-1)^{\frac{\xi}{\omega}}.
\end{split}
\end{equation}
Then we obtain
\begin{equation}
\frac{d}{dz}[G_{0}(z)e^{-\frac{\lambda}{\omega}z}(z-1)^{\frac{\xi}{\omega}%
}]=\frac{M}{\omega z-\omega}e^{-\frac{\lambda}{\omega}z}(z-1)^{\frac{\xi
}{\omega}}.
\end{equation}
Integrating from $0$ to $z$, we get%
\begin{equation}
\begin{split}
G_{0}(z)e^{-\frac{\lambda}{\omega}z}(z-1)^{\frac{\xi}{\omega}}&-G_{0}%
(0)=\\
&\frac{M}{\omega}\int_{0}^{z}e^{-\frac{\lambda}{\omega}x}(x-1)^{\frac{\xi
}{\omega}-1}dx.
\end{split}
\end{equation}
Further, we have%
\begin{equation}
\begin{split}
G_{0}(1)&=\lim_{z\longrightarrow1}G_{0}(z)\\
=&\lim_{z\longrightarrow1}\frac
{G_{0}(0)+\frac{M}{\omega}\int_{0}^{z}e^{-\frac{\lambda}{\omega}x}%
(x-1)^{\frac{\xi}{\omega}-1}dx}{e^{-\frac{\lambda}{\omega}z}(z-1)^{\frac
{\xi}{\omega}}}.%
\end{split}
\end{equation}
Due to $G_{0}(1)$ is a infinite value, we obtain that
\begin{equation}\label{}
  G_{0}(0)=-\frac{M}{\omega}\int_{0}^{z}e^{-\frac{\lambda}{\omega}x}(x-1)^{\frac{\xi}{\omega}-1}dx.
\end{equation}
For simplicity, we rewrite the equation as $G_{0}(0)=-\frac{M}{\omega}K$.
So that, $G_{0}(1)$ is obtained as
\begin{align}
G_{0}(1)  &  =\lim_{z\longrightarrow1}\frac{\frac{M}{\omega}K+\frac{M}{\omega
}\int_{0}^{z}e^{-\frac{\lambda}{\omega}x}(x-1)^{\frac{\xi}{\omega}-1}%
dx}{e^{-\frac{\lambda}{\omega}z}(z-1)^{\frac{\xi}{\omega}}}\\
&  =\lim_{z\longrightarrow1}\frac{\frac{M}{\omega}K(1-\frac{\int_{0}%
^{z}e^{-\frac{\lambda}{\omega}x}(x-1)^{\frac{\xi}{\omega}-1}dx}{\int_{0}%
^{1}e^{-\frac{\lambda}{\omega}x}(x-1)^{\frac{\xi}{\omega}-1}dx})}%
{e^{-\frac{\lambda}{\omega}z}(z-1)^{\frac{\xi}{\omega}}}.\nonumber%
\end{align}
By L'Hopital rule, we have%
\begin{align}
G_{0}(1)  &  =\frac{\frac{M}{\omega}K}{\frac{\xi}{\omega}\int_{0}%
^{1}e^{-\frac{\lambda}{\omega}x}(x-1)^{\frac{\xi}{\omega}-1}dx}\\
&  =\frac{\frac{M}{\omega}K}{\frac{\xi}{\omega}K}=\frac{M}{\xi}.\nonumber%
\end{align}
According to the equation of $G_{1}(z)$ and L'Hopital rule, we have%
\begin{align}
G_{1}(1)  &  =\lim_{z\longrightarrow1}\frac{\xi zG_{0}(z)-z\nu P_{11}-\xi
zP_{00}}{\lambda z+\nu z-\nu-\lambda z^{2}}\\
&  =\frac{\xi G_{0}(1)+\xi G_{0}^{\prime}(1)-\nu P_{11}-\xi P_{00}}%
{\nu-\lambda}.\nonumber%
\end{align}
Therefore, we have%
\begin{equation}
E[L_{0}]=G_{0}^{\prime}(1)=\frac{\nu-\lambda}{\xi}\sum_{n=1}^{\infty}P_{1n}.%
\end{equation}
According to $G_{0}^{\prime}(z)$, we obtian
\begin{equation}
E[L_{0}]=\lim_{z\longrightarrow1}G_{0}^{\prime}(z)=\frac{\lambda G_{0}(1)-\xi
G_{0}^{\prime}(1)}{\omega}.%
\end{equation}Due to $\sum_{n=1}^{\infty}P_{1n}+\sum_{n=0}^{\infty}P_{0n}=1$, we obtain
\begin{equation}
E[L_{0}]=\frac{\lambda\nu-\lambda^{2}}{\nu\omega+\nu\xi-\lambda\omega}.%
\end{equation}
Similarly, according to $G_{1}(z)$, we have
\begin{align}
E[L_{1}]  &  =\lim_{z\longrightarrow1}G_{1}^{\prime}(z)\\
&  =\lim_{z\longrightarrow1}\frac{\xi zG_{0}(z)-z\nu P_{11}-\xi zP_{00}%
}{\lambda z+\nu z-\nu-\lambda z^{2}}.\nonumber%
\end{align}
By L'Hopital rule, it becomes
\begin{align}
E[L_{1}]  &  =\frac{\xi G_{0}(1)+\xi G_{0}^{\prime}(1)-\nu P_{11}-\xi
P_{00}}{\nu-\lambda}\\
&  =\frac{\xi(\lambda\nu-\lambda^{2})}{(\nu\omega+\nu\xi-\lambda\omega
)(\nu-\lambda)}\\
&  =\frac{\xi\lambda}{\nu\omega+\nu\xi-\lambda\omega}.\nonumber%
\end{align}
The derivation for $E[L_{0}]$ and $E[L_{1}]$ is proved.
%\cite{joerer2014vehicular}\cite{chakchouk2015survey}\cite{sahu2013bahg}\cite{bitam2015bio}\cite{siegel2017survey}\cite{cao2017ogcmac}\cite{qiu2015methodology}\cite{dressler2017not}

%\bibliographystyle{IEEEtran}
%%\bibliography{Stablevehicle}

\begin{thebibliography}{1}
\bibitem{indexglobal}
C.~V.~N. Index, ``Global mobile data traffic forecast update, 2016--2021 white
  paper,'' accessed on May 2, 2017.

\bibitem{yu2017mobile}
H.~Yu, M.~H. Cheung, G.~Iosifidis, L.~Gao, L.~Tassiulas, and J.~Huang, ``Mobile
  data offloading for green wireless networks,'' \emph{IEEE Wireless
  Communications}, vol.~24, no.~4, pp. 31--37, 2017.

%\bibitem{wp5d2014report}
%I.~WP5D, ``Report ITU-R m. 2320-0,'' \emph{Future technology trends of
%  terrestrial IMT systems}, 2014.

\bibitem{chopra2018possible}
G.~Chopra, S.~Jain, and R.~K. Jha, ``Possible security attack modeling in
  ultradense networks using high-speed handover management,'' \emph{IEEE
  Transactions on Vehicular Technology}, vol.~67, no.~3, pp. 2178--2192, 2018.

\bibitem{Shi2016Edge}
W.~Shi, J.~Cao, Q.~Zhang, Y.~Li, and L.~Xu, ``Edge computing: Vision and
  challenges,'' \emph{IEEE Internet of Things Journal}, vol.~3, no.~5, pp.
  637--646, 2016.

\bibitem{hui2017content}
Y.~Hui, Z.~Su, T. H. ~Luan, and J.~Cai, ``Content in Motion: An Edge Computing
Based Relay Scheme for Content Dissemination in Urban Vehicular Networks,'' \emph{IEEE
  Transactions on Intelligent Transportation Systems}, 2018.

\bibitem{chen2018data}
M.~Chen, Y.~Qian, Y.~Hao, Y.~Li, and J.~Song, ``Data-driven computing and
  caching in 5G networks: Architecture and delay analysis,'' \emph{IEEE
  Wireless Communications}, vol.~25, no.~1, pp. 2--8, 2018.

\bibitem{jacobson2009networking}
V.~Jacobson, D.~K. Smetters, J.~D.~Thornton, M.~F.~Plass, N.~H.~Briggs,
  and R.~L.~Braynard, ``Networking named content,'' in \emph{ACM Proceedings of the
  5th international conference on Emerging networking experiments and
  technologies}, 2009, pp. 1--12.

\bibitem{zhu2016geographic}
L.~Zhu, C.~Li, B.~Li, X.~Wang, and G.~Mao, ``Geographic routing in multilevel
  scenarios of vehicular ad hoc networks,'' \emph{IEEE Transactions on
  Vehicular Technology}, vol.~65, no.~9, pp. 7740--7753, 2016.

\bibitem{zhu2015on}
L.~Zhu, C.~Li, Y. ~Wang, Z. ~Luo, Z. ~Liu, B. ~Li and X. ~Wang, ``On Stochastic Analysis of Greedy
Routing in Vehicular Networks,'' \emph{IEEE Transactions on
  Intelligent Transportation Systems}, vol.~16, no.~6, pp. 3353--3366, 2015.

\bibitem{wang2012multi}
Q. ~Wang, S. ~Leng, H. ~Fu, and Y. ~Zhang, ``An IEEE 802.11p-based Multi-channel MAC
Scheme with Channel Coordination for Vehicular Ad Hoc Networks,'' \emph{IEEE Transactions on
  Intelligent Transportation Systems}, vol.~13, no.~2, pp. 449--458, 2012.

\bibitem{li2018building}
C.~Li, Y.~Zhang, T. H.~Luan, and Y.~Fu, ``Building Transmission Backbone for Highway
Vehicular Networks: Framework and Analysis,'' \emph{IEEE Transactions on
  Vehicular Technology}, vol.~67, no.~9, pp. 8709--8722, 2018.

\bibitem{cheng2018bigdate}
N. ~Cheng, F. ~Lyu, J. ~Chen, W. ~Xu, H. ~Zhou, S. ~Zhang, and X. ~Shen,
``Big data driven vehicular networks,'' \emph{IEEE Network}, vol.~32, no.~6, pp. 160--167, 2018.


%\bibitem{wang2018shadowing}
%Q.~Wang, D.~W. Matolak, and B.~Ai, ``Shadowing characterization for 5-Ghz
%  vehicle-to-vehicle channels,'' \emph{IEEE Transactions on Vehicular
%  Technology}, vol.~67, no.~3, pp. 1855--1866, 2018.

%\bibitem{peng2017performance}
%H.~Peng, D.~Li, K.~Abboud, H.~Zhou, H.~Zhao, W.~Zhuang, and X.~S. Shen,
%  ``Performance analysis of ieee 802.11p {DCF} for multiplatooning
%  communications with autonomous vehicles,'' \emph{IEEE Transactions on
%  Vehicular Technology}, vol.~66, no.~3, pp. 2485--2498, 2017.

%\bibitem{peng2017resource}
%H.~Peng, D.~Li, Q.~Ye, K.~Abboud, H.~Zhao, W.~Zhuang, and X.~Shen, ``Resource
%  allocation for cellular-based inter-vehicle communications in autonomous
%  multiplatoons,'' \emph{IEEE Transactions on Vehicular Technology}, vol.~66,
%  no.~12, pp. 11\,249--11\,263, 2017.

\bibitem{luan2015social}
T.~H. Luan, R.~Lu, X.~Shen, and F.~Bai, ``Social on the road: Enabling secure
  and efficient social networking on highways,'' \emph{IEEE Wireless
  Communications}, vol.~22, no.~1, pp. 44--51, 2015.

\bibitem{fu2018infrastructure}
Y.~Fu, C.~Li, T. H. ~Luan, Y.~Zhang, and G.~Mao, ``Infrastructure-cooperative algorithm for effective intersection
collision avoidance,'' \emph{Transportation Research Part C: Emerging Technologies}, vol.~89,
 pp. 188--204, 2018.

\bibitem{liu2016energy}
D.~Liu and C.~Yang, ``Energy efficiency of downlink networks with caching at
  base stations,'' \emph{IEEE Journal on Selected Areas in Communications},
  vol.~34, no.~4, pp. 907--922, 2016.

\bibitem{gabry2016energy}
F.~Gabry, V.~Bioglio, and I.~Land, ``On energy-efficient edge caching in
  heterogeneous networks,'' \emph{IEEE Journal on Selected Areas in
  Communications}, vol.~34, no.~12, pp. 3288--3298, 2016.

\bibitem{ji2017order}
M.~Ji, A.~M. Tulino, J.~Llorca, and G.~Caire, ``Order-optimal rate of caching
  and coded multicasting with random demands,'' \emph{IEEE Transactions on
  Information Theory}, vol.~63, no.~6, pp. 3923--3949, 2017.

\bibitem{song2017learning}
J.~Song, M.~Sheng, T.~Q. Quek, C.~Xu, and X.~Wang, ``Learning-based content
  caching and sharing for wireless networks,'' \emph{IEEE Transactions on
  Communications}, vol.~65, no.~10, pp. 4309--4324, 2017.

%\bibitem{xu2018joint}
%J.~Xu, L.~Chen, and P.~Zhou, ``Joint service caching and task offloading for
%  mobile edge computing in dense networks,'' \emph{arXiv preprint
%  arXiv:1801.05868}, 2018.

\bibitem{zhao2018mobility}
Z.~Zhao, L. ~Guardalben, M. ~Karimzadeh, J. ~Silva, T. ~Braun and S. ~Sargento,
``Mobility prediction-assisted over-the-top edge prefetching for hierarchical VANETs,''
\emph{IEEE Journal on Selected Areas in Communications}, pp. 1--16, 2018.


\bibitem{vigneri2017per}
L.~Vigneri, S.~Pecoraro, T.~Spyropoulos, and C.~Barakat, ``Per chunk caching
  for video streaming from a vehicular cloud,'' in \emph{ACM MobiCom Workshop
  on Challenged Networks (CHANTS)}, 2017.

\bibitem{Vigneri2017Quality}
L.~Vigneri, T.~Spyropoulos, and C.~Barakat, ``Quality of experience-aware
  mobile edge caching through a vehicular cloud,'' in \emph{The ACM
  International Conference}, 2017, pp. 91--98.

%\bibitem{vigneri2016storage}
%L.~Vigneri, T.~Spyropoulos, and C.~Barakat, ``Storage on wheels: Offloading popular contents through a vehicular
%  cloud,'' in \emph{2016 IEEE 17th International Symposium on A World of Wireless, Mobile and Multimedia Networks
%  (WoWMoM),} 2016, pp. 1--9.

\bibitem{zhang2015interference}
R.~Zhang, X.~Cheng, L.~Yang, and B.~Jiao, ``Interference graph-based resource
  allocation (InGRA) for D2D communications underlaying cellular networks,''
  \emph{IEEE Transactions on Vehicular Technology}, vol.~64, no.~8, pp.
  3844--3850, 2015.

\bibitem{li2014throughput}
C.~Li, J.~Zhang, and K.~B. Letaief, ``Throughput and energy efficiency analysis
  of small cell networks with multi-antenna base stations,'' \emph{IEEE
  Transactions on Wireless Communications}, vol.~13, no.~5, pp. 2505--2517,
  2014.

\bibitem{karagiannis2010power}
T.~Karagiannis, J.-Y. Le~Boudec, and M.~Vojnovic, ``Power law and exponential
  decay of intercontact times between mobile devices,'' \emph{IEEE Transactions
  on Mobile Computing}, vol.~9, no.~10, pp. 1377--1390, 2010.

\bibitem{golrezaei2013femtocaching}
N.~Golrezaei, A.~F. Molisch, A.~G. Dimakis, and G.~Caire, ``Femtocaching and
  device-to-device collaboration: A new architecture for wireless video
  distribution,'' \emph{IEEE Communications Magazine}, vol.~51, no.~4, pp.
  142--149, 2013.

\bibitem{breslau1999web}
L.~Breslau, P.~Cao, L.~Fan, G.~Phillips, and S.~Shenker, ``Web caching and
  zipf-like distributions: Evidence and implications,'' in \emph{
  Eighteenth Annual Joint Conference of the IEEE Computer and Communications
  Societies (INFOCOM)}, 1999, pp. 126--134.

\bibitem{ren2015power}
Y.~Ren, F.~Liu, Z.~Liu, C.~Wang, and Y.~Ji, ``Power control in D2D-based
  vehicular communication networks,'' \emph{IEEE Transactions on Vehicular
  Technology}, vol.~64, no.~12, pp. 5547--5562, 2015.

\bibitem{liang2017resource}
L.~Liang, G.~Y. Li, and W.~Xu, ``Resource allocation for D2D-enabled vehicular
  communications,'' \emph{IEEE Transactions on Communications}, vol.~65, no.~7,
  pp. 3186--3197, 2017.

\bibitem{yang2017high}
H.~Yang, K.~Zheng, L.~Zhao, K.~Zhang, P.~Chatzimisios, and Y.~Teng, ``High
  reliability and low latency for vehicular networks: Challenges and
  solutions,'' \emph{arXiv preprint arXiv:1712.00537}, 2017.

\bibitem{liu2014will}
D.~Liu and C.~Yang, ``Will caching at base station improve energy efficiency of
  downlink transmission?'' in \emph{IEEE Global Conference on Signal and Information Processing
  (GlobalSIP),} 2014, pp. 173--177.

\bibitem{cheng2015d2d}
X.~Cheng, L.~Yang, and X.~Shen, ``D2D for intelligent transportation systems: A
  feasibility study,'' \emph{IEEE Transactions on Intelligent Transportation
  Systems}, vol.~16, no.~4, pp. 1784--1793, 2015.

\bibitem{chen2017throughput}
J.~Chen, G.~Mao, C.~Li, A.~Zafar, and A.~Y.~Zomaya, ``Throughput of
infrastructure-based cooperative vehicular networks,'' \emph{IEEE Transactions
on Intelligent Transportation Systems}, vol. 18, no. 11, pp. 2964--2979,
2017.

\bibitem{zhang2018real}
S.~Zhang, Y.~Luo, K.~Li, and V.~Li, ``Real-time energy-efficient control for
  fully electric vehicles based on explicit model predictive control method,''
  \emph{IEEE Transactions on Vehicular Technology}, 2018.

\bibitem{arnold2010power}
O.~Arnold, F.~Richter, G.~Fettweis, and O.~Blume, ``Power consumption modeling
  of different base station types in heterogeneous cellular networks,'' in
  \emph{IEEE Future Network and Mobile Summit, } 2010, pp. 1--8.

\bibitem{choi2012network}
N.~Choi, K.~Guan, D.~C. Kilper, and G.~Atkinson, ``In-network caching effect on
  optimal energy consumption in content-centric networking,'' in
  \emph{IEEE International Conference on Communications (ICC),} 2012, pp. 2889--2894.

\bibitem{altman2006analysis}
E.~Altman and U.~Yechiali, ``Analysis of customers¡¯ impatience in queues with
  server vacations,'' \emph{Queueing Systems}, vol.~52, no.~4, pp. 261--279,
  2006.

\bibitem{perel2010queues}
N.~Perel and U.~Yechiali, ``Queues with slow servers and impatient customers,''
  \emph{European Journal of Operational Research}, vol. 201, no.~1, pp.
  247--258, 2010.

\bibitem{3rdgeneration2016project}
3rd Generation Partnership Project: \emph{Technical Specification Group
Radio Access Network: Study LTE-Based V2X Services: (Release 14)},
Standard 3GPP TR 36.885 V2.0.0, Jun. 2016.

\bibitem{zhang2018social}
H.~Zhang, Z.~Wang, and Q.~Du, ``Social-aware D2D relay networks for stability
enhancement: an optimal stopping approach,'' \emph{IEEE Transactions on Vehicular
Technology}, vol.~67, no.~9, pp. 8860--8874, 2018.

\bibitem{dinkelbach1967nonlinear}
W.~Dinkelbach, ``On nonlinear fractional programming,'' \emph{Management
  science}, vol.~13, no.~7, pp. 492--498, 1967.

\bibitem{neely2010stochastic}
M.~J. Neely, ``Stochastic network optimization with application to
  communication and queueing systems,'' \emph{Synthesis Lectures on
  Communication Networks}, vol.~3, no.~1, pp. 1--211, 2010.

\bibitem{neely2013dynamic}
M.~J. Neely, ``Dynamic optimization and learning for renewal systems,'' \emph{IEEE
  Transactions on Automatic Control}, vol.~58, no.~1, pp. 32--46, 2013.

\bibitem{neely2006energy}
M.~J. Neely, ``Energy optimal control for time-varying wireless networks,''
  \emph{IEEE transactions on Information Theory}, vol.~52, no.~7, pp.
  2915--2934, 2006.

\bibitem{jia2016fusing}
Y.~Jia, X.~Song, J.~Zhou, L.~Liu, L.~Nie, and D.~S. Rosenblum, ``Fusing social
  networks with deep learning for volunteerism tendency prediction.'' in
  \emph{AAAI}, 2016, pp. 165--171.

%\bibitem{liu2012renewable}
%Z.~Liu, Y.~Chen, C.~Bash, A.~Wierman, D.~Gmach, Z.~Wang, M.~Marwah, and
%  C.~Hyser, ``Renewable and cooling aware workload management for sustainable
%  data centers,'' in \emph{ACM SIGMETRICS Performance Evaluation Review}, 2012, pp.
%  175--186.

\bibitem{roma-taxi-20140717}
L.~Bracciale, M.~Bonola, P.~Loreti, G.~Bianchi, R.~Amici, and A.~Rabuffi,
  ``{CRAWDAD} dataset roma/taxi (v. 2014-07-17),'' Downloaded from
  \url{https://crawdad.org/roma/taxi/20140717}, Jul. 2014.

\end{thebibliography}

%\appendices
%\section{Proof of the ...}
%\section*{Acknowledgment}
%\ifCLASSOPTIONcaptionsoff
%  \newpage
%\fi

\begin{IEEEbiography}[{\includegraphics[width=1.4in,height=1.2in,clip,keepaspectratio]{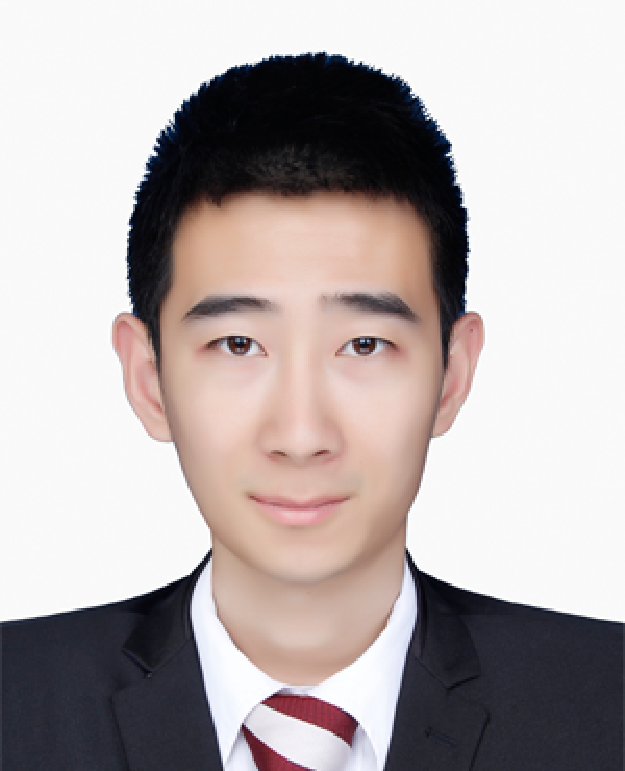}}]
{Yao Zhang}received the B.Eng. degree in Telecommunication Engineering from Xi¡¯an University of Science and Technology, China, in 2015, and is currently pursuing the Ph.D. degree in Telecommunication Engineering, at Xidian University, Xi¡¯an, China. His current research interests include communication protocol and performance
evaluation of vehicular networks, edge caching, and wireless sensor networks.
\end{IEEEbiography}
\vspace{-4.2em}
\begin{IEEEbiography}[{\includegraphics[width=1in,height=1.2in,clip,keepaspectratio]{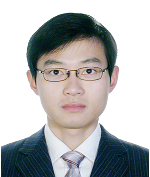}}]
{Changle Li} received the Ph.D. degree in communication and information system from Xidian University, China, in 2005. He conducted his postdoctoral research in Canada and the National Institute of information and Communications Technology, Japan, respectively. He had been a Visiting Scholar with the University of Technology Sydney and is currently a Professor with the State Key Laboratory of Integrated Services Networks, Xidian University. His research interests include intelligent transportation systems, vehicular networks, mobile ad hoc networks, and wireless sensor networks.
\end{IEEEbiography}
\vspace{-4.2em}
\begin{IEEEbiography}[{\includegraphics[width=1in,height=1.2in,clip,keepaspectratio]{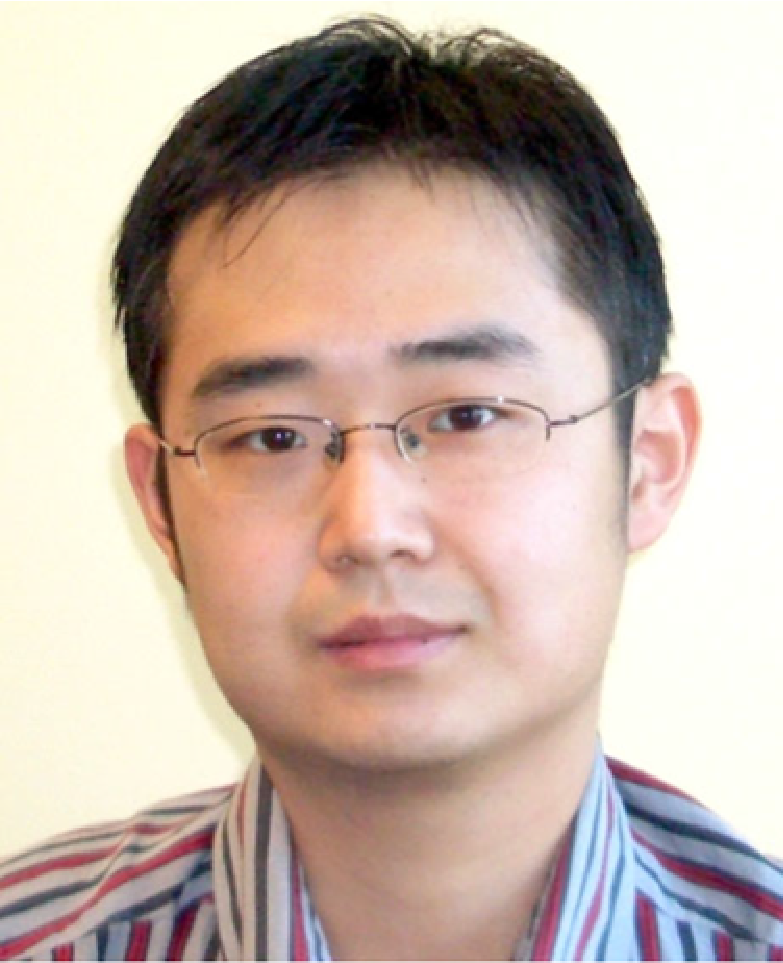}}]
{Tom H. Luan} received his B.Eng. degree from Xi'an Jiao Tong University, China, in 2004, the M.Phil. degree from Hong Kong University of Science and Technology in 2007, and Ph.D. degree from the University of Waterloo, Ontario, Canada, in 2012. He is a professor at the School of Cyber Engineering of Xidian University, Xi'an, China. His research mainly focuses on content distribution and media streaming in vehicular ad hoc networks and peer-to-peer networking, as well as the protocol design and performance evaluation of wireless cloud computing and edge computing. Dr. Luan has authored/coauthored more than 40 journal papers and 30 technical papers in conference proceedings, and awarded one US patent. He served as a TPC member for IEEE Globecom, ICC, PIMRC and the technical reviewer for multiple IEEE Transactions including TMC, TPDS, TVT, TWC and ITS.
\end{IEEEbiography}
\vspace{-4.2em}
\begin{IEEEbiography}[{\includegraphics[width=1in,height=1.33in,clip,keepaspectratio]{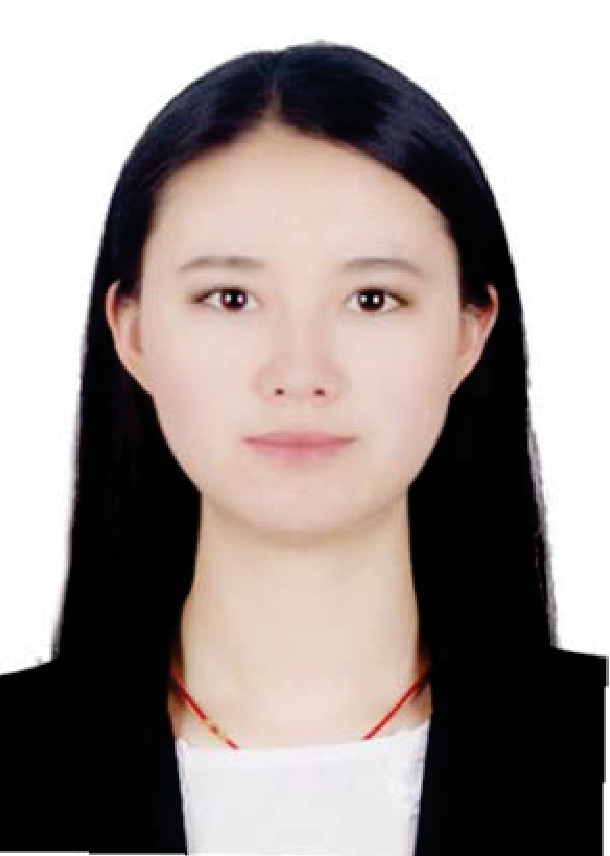}}]
{Yuchuan Fu}received the B.Eng. degree in Telecommunication Engineering from Xi'an University of Posts \& Telecommunications, China, in 2014, and is currently pursuing the Ph.D. degree in Telecommunication Engineering, at Xidian University, Xi'an, China. Her current research interests include communication protocol and algorithm design in vehicular networks and wireless sensor networks.
\end{IEEEbiography}
\vspace{-3.2em}
\begin{IEEEbiography}[{\includegraphics[width=1.5in,height=1.33in,clip,keepaspectratio]{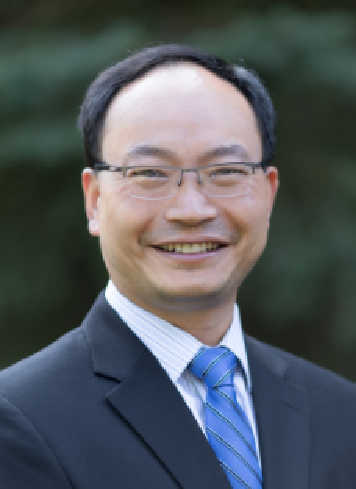}}]
{Weisong Shi} received the BS degree from Xidian University, in 1995, and the PhD degree from the Chinese Academy of Sciences, in 2000, both in computer engineering. He is a
Charles H. Gershenson distinguished faculty fellow and a professor of computer science with Wayne State University. His research interests include edge computing, computer
systems, energy-efficiency, and wireless health. He is a recipient of the National Outstanding PhD dissertation award of China and the NSF CAREER award. He is a fellow of the
IEEE and ACM distinguished scientist.
\end{IEEEbiography}
\vspace{-4em}
\begin{IEEEbiography}[{\includegraphics[width=1in,height=1.32in,clip,keepaspectratio]{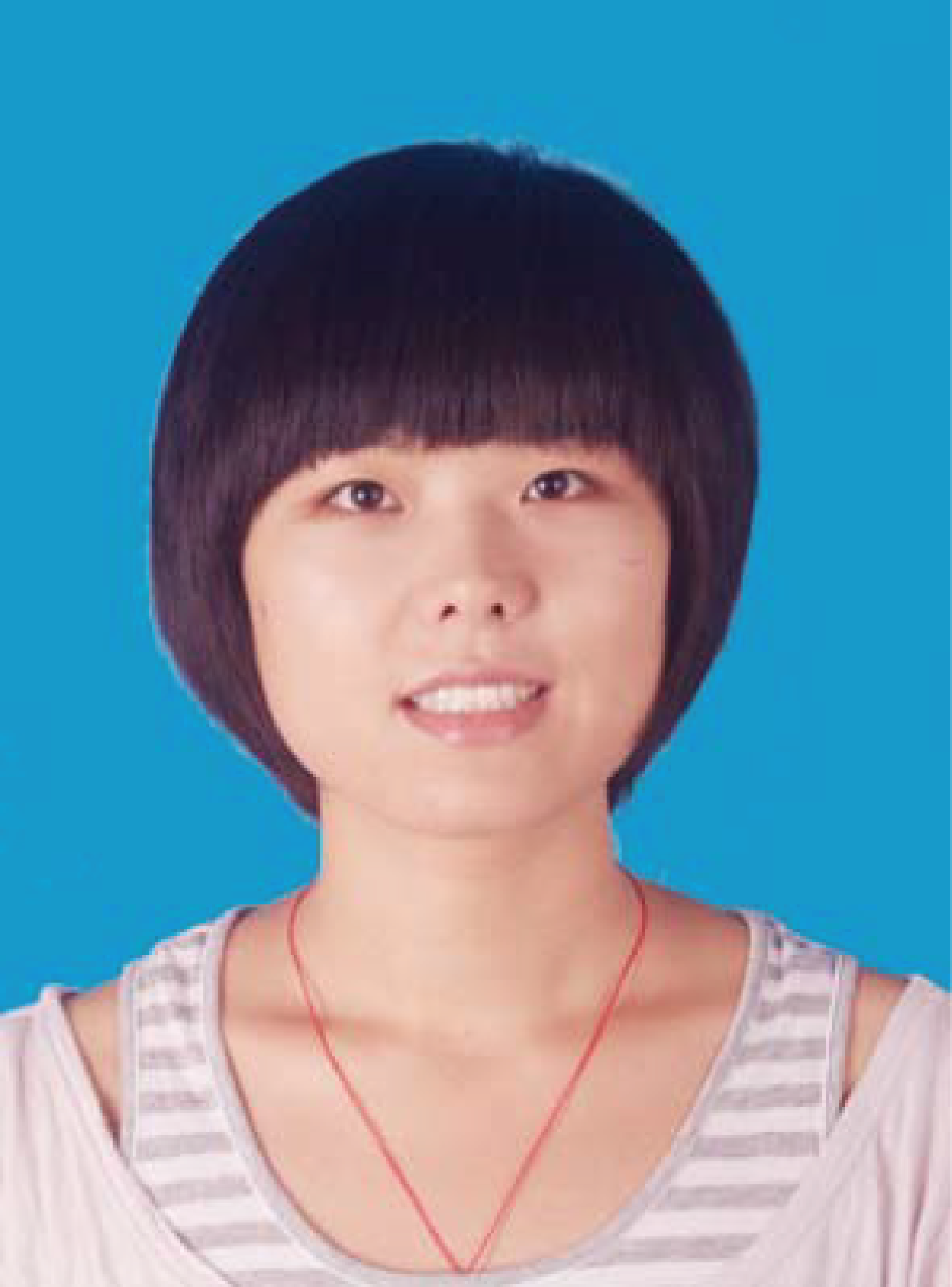}}]
{Lina Zhu} received her B.E. degree from Suzhou University of Science and Technology, China, in 2009, and Ph.D. degrees in Communication and Information System, from Xidian University, China, in 2015. She is currently a lecturer in State Key Laboratory of Integrated Services Networks at Xidian University, China. Her current research interests include mobility model,  trust forwarding, routing and MAC protocols in vehicular networks.
\end{IEEEbiography}
%It is not necessary to upload the biography when you submit your manuscript.

\end{document}